\documentclass[a4paper]{article}

\usepackage[margin=3cm]{geometry}
\RequirePackage{amsthm,amsmath,amsfonts,amssymb, bbm}
\RequirePackage[numbers]{natbib}
\RequirePackage[colorlinks,citecolor=blue,urlcolor=blue]{hyperref}
\RequirePackage{graphicx}
\usepackage{mathtools}
\usepackage{dsfont}
\usepackage{etoolbox}
\apptocmd{\thebibliography}{\raggedright}{}{}
\usepackage{listings}
\usepackage{enumerate}
\usepackage{enumitem}

\theoremstyle{plain}
\newtheorem{theorem}{Theorem}[section]

\newtheorem{lemma}[theorem]{Lemma}
\newtheorem{proposition}[theorem]{Proposition}

\theoremstyle{remark}
\newtheorem{definition}[theorem]{Definition}
\newtheorem{example}[theorem]{Example}
\newtheorem{remark}[theorem]{Remark}
\newtheorem{assumption}[theorem]{Assumption}

\newcommand{\whW}{\widehat{W}}

\newcommand{\bC}{\mathbb{C}}
\newcommand{\bE}{\mathbb{E}}
\newcommand{\bN}{\mathbb{N}}

\newcommand{\bR}{\mathbb{R}}

\newcommand{\cE}{\mathcal{E}}
\newcommand{\cF}{\mathcal{F}}

\newcommand{\cI}{\mathcal{I}}
\newcommand{\cJ}{\mathcal{J}}

\newcommand{\cL}{\mathcal{L}}

\newcommand{\cW}{\mathcal{W}}

\newcommand{\fA}{\mathfrak{A}}
\newcommand{\fB}{\mathfrak{B}}
\newcommand{\fC}{\mathfrak{C}}
\newcommand{\fD}{\mathfrak{D}}
\newcommand{\fE}{\mathfrak{E}}

\newcommand{\fI}{\mathfrak{I}}
\newcommand{\fJ}{\mathfrak{J}}

\newcommand{\fL}{\mathfrak{L}}
\newcommand{\fM}{\mathfrak{M}}

\newcommand{\wI}{I}
\newcommand{\wJ}{J}

\begin{document}
\title{Weak error estimates for rough volatility models}
\date{\today}

\author{Peter K. Friz\thanks{Weierstrass Institute for Analysis and Stochastics and TU Berlin; \href{friz@math.tu-berlin.de}{friz@math.tu-berlin.de}. Supported by the Deutsche Forschungsgemeinschaft (DFG) under Germany's Excellence Strategy -- The Berlin Mathematics Research Center MATH+ (EXC-2046/1, project ID: 390685689) and also by DFG CRC/TRR 388 ``Rough Analysis, Stochastic Dynamics and Related Topics''.} \qquad William Salkeld \thanks{University of Nottingham; \href{william.salkeld@nottingham.ac.uk}{william.salkeld@nottingham.ac.uk}.
Supported by MATH+ project AA4-2, while this project was initiated, and then by the US Office of Naval Research under the Vannevar Bush Faculty Fellowship N0014-21-1-2887.
}
\qquad Thomas Wagenhofer\thanks{TU Berlin; \href{mailto:wagenhof@math.tu-berlin.de}{wagenhof@math.tu-berlin.de}.
Initially supported by MATH+, as PhD student in the Berlin Mathematical School (BMS), and then by IRTG 2544 ”Stochastic Analysis in Interaction” -  DFG project-ID 410208580.
}}

    \makeatletter
    \def\@maketitle{%
        \newpage
        \null
        \vskip 0.1em%
        \begin{center}%
            \let \footnote \thanks
            {\LARGE \@title \par}%
            \vskip 1.5em%
                {\large
            \lineskip .5em%
            \begin{tabular}[t]{c}%
                \@author
            \end{tabular}\par}%
            \vskip 1em%
                {\large \@date}%
        \end{center}%
        \par
        \vskip 1.5em}
    \makeatother

\maketitle

\begin{abstract}
We consider a class of stochastic processes with rough stochastic volatility, examples of which include the rough Bergomi and rough Stein-Stein model, that have gained considerable importance in quantitative finance. 
			
A basic question for such (non-Markovian) models concerns efficient numerical schemes. While strong rates are well understood (order $H$), we tackle here the intricate question of weak rates. Our main result asserts that the weak rate, for a reasonably large class of test function, is essentially of order $\min \{ 3H+\tfrac12, 1 \}$ where $H \in (0,1/2]$ is the Hurst parameter of the fractional Brownian motion that underlies the rough volatility process. 
			
Interestingly, the phase transition at $H=1/6$ is related to the correlation between the two driving factors, and thus gives additional meaning to a quantity already of central importance in stochastic volatility modelling.
Our results are complemented by a lower bound which show that the obtained weak rate is indeed optimal. 
\end{abstract}
\vspace{0.3cm}


\section{Introduction}
	Recall some standard results from the numerics of stochastic differential equations:  Euler left-point approximations (a.k.a.\ the Euler-Maruyama scheme) with step size $1/n$, say $X^{(n)}$, converge to the limit $X$, described by an It\^{o} integral equation, with $L^2$-rate $\sqrt{ E ( | X - X^{(n)} |^2 )} = O(n^{-1/2}) $. This is known as {\em strong rate} $1/2$, which is essentially a consequence of $\sqrt{E( |W_t - W_s|^2)} = |t-s|^{1/2}$, where $W$ denotes a standard Brownian motion $W$. On the other hand, typically more relevant in practice, one has {\em weak rate} $1$ - twice the strong rate - meaning that
	\begin{equation*}
		\bE\Bigl[ \Phi \bigl( X_T \bigl) \Bigr] - \bE\Bigl[ \Phi  \bigl( X_T^{(n)} \bigl) \Bigr] = O \bigl( n^{-1} \bigl)
	\end{equation*}
	for $\Phi$ in a suitable class of test functions. This is found in many textbooks, a classical reference is \cite{Talay1990ExpansionOT}. 
	
	\medskip
	For many years now, authors have studied stochastic systems which involve a fractional Brownian motion (fBm). A key property of such a Gaussian process $\whW$ is the fractional scaling ${E( |\whW_t|^2)} \propto |t|^{2H}$. The exponent $H \in (0,1)$, known as Hurst parameter, determines the sample path roughness.
	Differential equations driven by fBm have been studied extensively by means of Young or rough integration theory, the pathwise nature of which is well-suited to a.s.\ strong rates (see e.g.\ \cite{friz2014convergence, liu2019first} where an optimal strong rate of $2H - 1/2$ is found, related to a phase transition at $H=1/4$ concerning the existence of fractional L\'evy area).
	\medskip
	
	The purpose of this article is to study weak rates for stochastic integrals of functions of fractional Brownian motion. These stochastic processes are for example rooted in quantitative finance  \cite{bayer2023rough}, \cite{Gatheral2014VolatilityIR} and of increasing popularity in both industry and academia. Mathematically, we are interested here in a (continuous) martingale $X$, interpreted as continuous asset price process, or stochastic logarithm thereof. The absolutely continuous characteristics of $X$, such as to have a well-defined stochastic variance (resp. volatility), are defined as $d \langle X \rangle / dt $ (resp. the square-root thereof), which is  assumed, on small time-scales, to exhibit fractional scaling with some Hurst parameter $H$. Throughout this paper, we focus on the ``rough'' regime $H\in (0,1/2]$, the fundamental importance of which \cite{Gatheral2014VolatilityIR} has been confirmed and (re)discussed  by a number of authors, see for example \cite{fukasawa2021volatility} and the references therein. Following  \cite{Bayer2015PricingUR,Gatheral2014VolatilityIR} a simple specification of such a process with rough volatility is given by  
	\begin{equation}
		\label{eq:RoughVolatilityModel0}
		d X_t = \sigma_t \,dB_t,
	\end{equation}
	with explicit rough volatility process 
	\begin{equation*}
		\sigma_t = f(\whW_t)\ , \quad \whW_t = \int_0^t (t-s)^{H-1/2} dW_s,
	\end{equation*}
	where $W$ and $B$ are correlated Brownian motions, more specifically for independent Brownian motions $W_t$ and $W_t^\perp$
	\begin{equation}
		\label{eq:BwithRHO}
		B_t= \rho W_t + \sqrt{1-\rho^2} {W}^\perp_t, \qquad \rho \in [-1,1].
	\end{equation}
	Here $f$ is a deterministic function, sometimes called volatility function, and one has, explicitly,
	\begin{equation}
		\label{eq:RoughVolatilityModelNew}
		X_t = X_0 + \int_0^t f\Bigl( \whW_s \Bigr) dB_s.
	\end{equation}
	Despite (or because) of its simplicity, this setting accomodates popular models including the {\em rough Stein-Stein} model (cf. \cite{jaber2020characteristic, Bayer2020weak}) and the {\em rough Bergomi} model introduced in \cite{Bayer2015PricingUR} with, respectively,  
	\begin{equation}
		\label{eq:f_needed}
		f(x) \in \{ c_1 x, c_2 \exp ( c_3 x) \}.
	\end{equation}
 The interest in studying these objects goes beyond mathematical finance.
Indeed, equation \eqref{eq:RoughVolatilityModelNew} can be viewed as model case of a stochastic
system with distinct time scales, through the simultaneous influence of Brownian and fractional Brownian noise. Including
also (nice) drift terms, which constitute a harmless perturbation from a weak
rate perspective, \eqref{eq:RoughVolatilityModelNew} embeds in multivariate stochastic differential
systems of the form
\begin{align*}
  X_t & = x_0 + \int_0^t \sigma_X (X_s, Y_s)  \hspace{0.17em} d \whW_s +
  \int_0^t \mu_X (X_s, Y_s)  \hspace{0.17em} ds,\\
  Y_t & = y_0 + \int_0^t \sigma_Y (X_s, Y_s)  \hspace{0.17em} dW_s + \int_0^t
  \mu_Y  (X_s, Y_s)  \hspace{0.17em} ds
\end{align*}
in dimension $d_X + d_Y$, with multidimensional fractional resp. classical
Brownian noise, with prescribed correlation structure. In dimension $1 + 1$,
take $\sigma_X \equiv 1, \mu_X \equiv 0, \sigma_Y (x, y) = \sigma (x)$ and
$\rho$-correlated $W$ and $\widehat{W}$. The question of weak rates for such
systems is largely open; our work can seen as a precise contribution to this problem. It also underlines the subtlety of the problem, for instance with regards to the correlation structure.

We note that such equations have been studied recently by a number of authors
in the context of fast-slow systems \cite{BourguinGailusSpiliopoulos2021SlowFast,HairerLi2022GeneratingDiffusions,HongLiLiu2022SlowFastMKV,LiSieber2022AveragingFbm,PeiInahamaXuFastSlowRoughPath,RocknerXie2021AveragingPrinciple} \ in which case all
coefficient fields scale with a homogenization parameter $\varepsilon$. It
would of course be interesting to explore the interplay of $\varepsilon$ with the Euler step-size parameter $n$, but this is very much beyond the scope of the present work.
	Throughout, we consider the standard left-point approximation of \eqref{eq:RoughVolatilityModelNew}, that is 
	\begin{align}
		\label{eq:RoughVolatilityModel-Discrete}
		X_t^{(n)} - X_0 &= \int_0^t f\Bigl( \whW_{\eta(s)} \Bigr) dB_s = \bigg( \sum_{i=0}^{\lfloor nt\rfloor -1} f\Bigl( \whW_{\tfrac{i}{n}} \Bigr) B_{\tfrac{i}{n}, \tfrac{i+1}{n}} \bigg) + f\Bigl( \whW_{\eta(t)} \Bigr) B_{\eta(t), t},
	\end{align}
    {where $\eta(s) = \tfrac{\lfloor n s\rfloor}{n}$ and $B_{s,t}=B_t-B_s$.
	One easily checks, for reasonable $f$, that $X^n \to X$ with strong rate $H$, uniformly on compacts in time. 
    We draw the readers attention to \cite{Nourdin2010CLTforSkorokhod} Section 5.1 where the authors consider similar summations of smooth functions of a fractional Brownian motion multiplied by increments of a Brownian motion converging stably to an appropriate stochastic integral.
    For rate $H$ close to zero, as suggested in most works on rough volatility, numerical simulation of such models may appear to be difficult. However, in many situations, including option pricing under such models, the weak rate matters. 
	The naive guess of weak rate $2H$, twice the strong rate, is not supported by numerical simulations - after all, option pricing under rough volatility works numerically surprisingly well.
	
	A second guess for the weak rate might then be $H+1/2$, taking into account the ``mixed'' appearance of $\whW$ and $B$, with their respective scaling exponents, in \eqref{eq:RoughVolatilityModelNew}. 
	Surprisingly perhaps, at least before Gassiat's work \cite{Gassiat2022Weak}, both guesses are wrong.
	In presence of correlation arising from \eqref{eq:BwithRHO}, the correct weak rate turns out to be
    $$
    \min (3H+1/2,1). 
    $$
    We establish this rate of convergence for a reasonably generic class of test functions $\Phi$ and volatility functions $f$, together with a lower bound (Section \ref{sec:LB}) that implies optimality. When $\rho$ is taken to be equal to 0, the weak rate improves to $1$. Either way, the rate always stays above $1/2$, which is particularly useful when dealing with Hurst parameter close to $0$. We have
	
	\begin{theorem}  
		\label{Thm_Rate}
		Let $X$ and $X^n$ be given by \eqref{eq:RoughVolatilityModelNew} and \eqref{eq:RoughVolatilityModel-Discrete}, respectively, with $f \in \mathcal{C}^{N}$, $N \in \mathbb{N}$, such that $f$ and its $N$ derivatives have at most exponential growth with constants $C_f'$ and $C_f$, uniformly among $f$ and its $N$ derivatives. 
		\begin{enumerate}[label=(\roman*)]
		    \item 
		    \label{Thm1i}
		    For any polynomial test function $\Phi$ with $\deg (\Phi) \le N$, there is a constant $C_N$ depending only on the coefficients of the polynomial such that, as $n\to \infty$,
    		\begin{equation*}
    			\bE\Bigl[ \Phi \bigl( X_T \bigl) \Bigr] - \bE\Bigl[ \Phi  \bigl( X_T^{(n)} \bigl) \Bigr] 
    			\le
    			\begin{cases}
    				C_N n^{-3H-1/2} \vee n^{-1} & \text{for } H \not = 1/6,
    				\\
    				C_N n^{-1}\log(n) & \text{for } H  = 1/6.
    			\end{cases}
    		\end{equation*}
    		\item 
    		\label{Thm1ii}
    		In the uncorrelated case, with $\rho = 0$ in \eqref{eq:BwithRHO}, we have
    		\begin{equation*}
    			\bE\Bigl[ \Phi \bigl( X_T \bigl) \Bigr] - \bE\Bigl[ \Phi  \bigl( X_T^{(n)} \bigl) \Bigr] 
    			\lesssim 
    			n^{-1}  \qquad \text{any } H>0.
    		\end{equation*}
    	\end{enumerate}
	\end{theorem}

    \begin{remark}\label{Rem_Rate}
        The proof of Theorem \ref{Thm_Rate} yields that $C_N$ can be chosen as
      \begin{equation*}
          C_N=
          C
          8^N N^3\bigl(N!\bigr)^2 \bigl(C'_f\bigr)^{2N}
          \exp\Big({2(C_f \cdot N)^2} \cdot \frac{T^{2H+1}}{(2H+1/2)^2} \Big) 
      \end{equation*}
      for a constant $C$ that does not depend on $f,N,T$ or $h$.
    \end{remark}

	The proof is surprisingly involved and seems to be far from classically used methods. The basic idea is an exact expansion, based on an iteration scheme, which terminates for polynomial test functions, using a mixture of It\^{o} and Malliavin calculus, leading to complicated iterated integrals that require both algebraic and analytic understanding. The left-point approximation admits a similar expansion, but with suitably modified kernel; the rate is then obtained from a term-by-term comparison. The origin of the rate $3H+1/2$, unintuitive at first, can be traced to the estimation of
	integrals of the form $\int_{t-2/n}^t \bigl(\varphi(s)-\varphi(\eta(s))\bigr) (t-s)^{H-1/2} \, ds$, which in fact appear in different places, e.g. \eqref{Eq_FormRate1} or \eqref{Eq_FormRate2}, where $\varphi$ is a $2H$-H\"older regular function, that arises from the expectation of fractional Brownian functionals.
	Optimality is settled by a lower bound in the case of $f(x)=x, \Phi(x) = x^3, \rho = \pm 1$. In case of $\Phi(x)=x^2$ one always has weak rate $1$, for any $H>0$, as one can see using It\^{o}'s isometry. In \cite{Bayer2020weak} this elegant observation is attributed to A. Neuenkirch. 
	
	\medskip 
	
	Over the last years, several authors \cite{Bayer2022weak, Bayer2020weak, Gassiat2022Weak} have studied this problem, albeit only in the case $\rho = 1$ in which case \eqref{eq:RoughVolatilityModelNew} simplifies to $d X =  f (\widehat W) d W$. 
	In \cite{Bayer2020weak}, the authors employ Markovian approximations to fBM through a family of Ornstein-Uhlenbeck processes, enabling them to recycle results from a diffusions setting. They are able to pass to the limit and obtain, for linear volatility function $f(x)=x$ and sufficiently nice (bounded) test functions, 
	a weak rate of (at least) $H+1/2$. 
	This is further improved by \cite{Bayer2022weak} with Malliavin calculus methods, yielding the same weak rate, still for linear $f(x)=x$, but now allowing polynomial growth test functions, further improved by \cite{Gassiat2022Weak}. Much less is know for non-linear $f$, despite a clear need in applications such as in the rough Bergomi model, cf. \eqref{eq:f_needed}. The only available result, to our knowledge, is due to \cite{Gassiat2022Weak} where the same rate is obtained as in Theorem \ref{Thm_Rate}, part \ref{Thm1i}, but only for cubic test functions $\Phi(x) =x^3$, generic $f$ and $\rho =1$.
    
    In the special case of fractional Brownian motion with Hurst parameter $H=1/6$, the question of defining the Stratonovich stochastic integral was studied in \cite{Nourdin2010weakStratonovic}, where key convergence results rely on the choice of $H$. In a recent preprint \cite{Bonesini2023WeakRates} the authors give a weak rate of $H+1/2$ for a relatively broad class of test functions.

	This paper addresses the case of general polynomial test functions, generic $f$ and arbitrary $\rho \in [-1,1]$. As we illustrate throughout with Examples \ref{example:Gassiat-X^4}, \ref{Ex_MomRep2} and Subsection \ref{subsec:4-Examples}, already the case $\Phi(x) =x^4$ leads to complicated iterated singular integrals. Part of the work is thus to find an efficient algebraic representation, followed by a fine analysis (term-by-term comparison) that ultimately leads to the weak rate. Our moment formula for $\mathbb{E}(X_T^N)$ (cf. Theorem \ref{Thm_MomRep}) has some similarity (in spirit) with a moment formula of \cite[Sec 4.2]{HairerLi2022GeneratingDiffusions}, in the context of fractional fast-slow systems. Another intriguing remark is that we deal with
	multivariate fractional correlation functions that are reminiscent of expressions seen in quantum field theory and renormalization theory of singular SPDEs. On a detailed technical level, after circulating a draft version of this paper, Ajay Chandra kindly pointed us to Chapter 9 of \cite{glimm2012quantum} where formula (9.1.33) therein can be seen as variation of Lemma \ref{Lem_DerivSigma}, providing unexpected connection between the mathematics of quantum field theory (QFT) and the numerical analysis of rough volatility models.
    Our referee kindly pointed out that the proof of Lemma \ref{Lem_DerivSigma} and Lemma \ref{Lem_DerivSigma2} can be found in \cite[Theorem 2.1]{Hirsch2011ItoforGauss}. Additionally, Felix Otto pointed us to \cite{ClozenJosienOttoXu2023Bias}, as well as \cite{Plackett1954ReductionFormula} and \cite{Price1958GaussianInputs} which also include similar results.

	\subsection{Notations and other preliminaries} 
	
	Let $T>0$ so that $[0,T]$ is a closed, bounded interval of the real line. Let 
	\begin{equation}
		\label{eq:Simplex}
		\Delta_m^\circ = \Big\{ (t_1, ..., t_m) \in [0,T]^{\times m}: 0<t_m<t_{m-1}<\dots<t_1<T \}. 
	\end{equation}
	We will often denote $\mathbf{t} = (t_1, ..., t_m) \in \Delta_m^{\circ}$. For any $m \in \mathbb{N}$ and some function $f:\Delta_m^\circ \to \mathbb{R}$, we extend the definition $f: [0,T]^{\times m} \to \bR$ by
	\begin{equation*}
		f(\mathbf{t}) = 0  \quad \text{ for } \mathbf{t}=(t_1,\dots,t_m)\not \in \Delta_m^\circ.
	\end{equation*}
	
	Throughout this paper, we denote the \emph{Liouville kernel} $K: \Delta_2^{\circ} \to \bR$ defined by
	\begin{equation}
		\label{Liouville-Kernel}
		K(t, s) = 
		\begin{cases}
			(t-s)^{H-1/2} \quad &\quad \mbox{if $t>s$}
			\\
			0\quad &\quad \mbox{if $t\leq s$}
		\end{cases}
	\end{equation}
	
	We will often use the convention that $t_0$ is an abstract variable and extend $K$ such that $K(t_0,\cdot)\equiv 1$.
	\begin{definition}
		\label{Def_ExpGrowth}
		Let $g:\Delta_m^{\circ}\rightarrow \mathbb{R}$ be in $C^m$. We say that $g$ has at most exponential growth, if for all $j,k\in \{1,\dots,m\}$ there exist constants $C_g,C_g'$ such that
		\begin{equation}
			\label{eq:Def_ExpGrowth}
			|g|+|\partial_j g|+|\partial_j\partial_k g|\le C_g' \exp(C_g|x_1|+\dots+C_g|x_m|).
		\end{equation}
	\end{definition}
	
	\section{Gaussian Computations}
	
	The following lemma will prove very useful. Remark that $ \partial_t\Sigma(t)$ below is not assumed to be positive semi-definite (example: Brownian bridge!) so that the following result, for which we offer a Fourier proof, does not seem to approachable by It\^o's formula. The following two Lemmas can also be found in \cite[Theorem 2.1]{Hirsch2011ItoforGauss} as was kindly pointed out by a referee. For the reader's convenience we also state the proof here. 
	\begin{lemma}
		\label{Lem_DerivSigma}
		Let $d\in \mathbb{N}$, $\Sigma: [0,T]\rightarrow \mathbb{R}^{d \times d}$ a continuously differentiable, matrix valued map, such that $\Sigma(t)$ is symmetric, positive semi-definite for all $t \in [0,T]$. 	Let $g:\mathbb{R}^d \rightarrow \mathbb{R}$ be a Schwartz function, i.e. smooth with rapidly decreasing derivatives.
		Let $W(t) \sim \mathcal{N}(0,\Sigma(t))$ and define the function
		\begin{gather*}
			\varphi(t)=\bE\bigl[ g(W(t)) \bigr].
		\end{gather*}
		Then $\varphi$ is in $C^1$ and
		\begin{gather*}
			\partial_t \varphi(t) = \sum_{k,l=1}^d \frac 12 \partial_t\Sigma(t)_{k,l} \bE \bigl[\partial_k \partial_l g(W(t))\bigr].
		\end{gather*}
	\end{lemma}
	
	\begin{proof} 
		Let $\hat g$ be the Fourrier-transform of $g$, i.e.\ 
		\begin{gather*}
			\hat g (\xi)=\frac 1 {\sqrt{2\pi}^d} \int_{\mathbb{R}^d} e^{\mathrm{i}\xi \cdot x} g(x) \, dx.
		\end{gather*}
		Note that 
		\begin{gather*}
			g(x)=\frac{1}{\sqrt{2 \pi}^d} \int_{\mathbb{R}^d} e^{-\mathrm{i}x \cdot \xi} \hat g(\xi) \, d\xi.
		\end{gather*}
		
		Then it follows from Fubini's theorem that
		\begin{gather*}
			\begin{aligned}
				\bE[g(W(t))]&=\frac 1 {\sqrt{2\pi}^d} \int_{\mathbb{R}^d} \bE\bigl[ e^{\mathrm{i}\xi \cdot W(t)} \bigr]\hat g(\xi) \, d\xi
				\\
				&=\frac 1 {\sqrt{2\pi}^d} \int_{\mathbb{R}^d} e^{-\frac 12  \xi \cdot\Sigma(t) \xi}\hat g(\xi) \, d\xi.
			\end{aligned}
		\end{gather*}
		One easily justifies that
		\begin{gather*}
			\begin{aligned}
				\partial_t \varphi(t)&=\frac 1 {\sqrt{2\pi}^d} \int_{\mathbb{R}^d} -\frac 12  \xi \partial_t\Sigma(t) \xi e^{-\frac 12  \xi \cdot \Sigma(t) \xi}\hat g(\xi) \, d\xi
				\\
				&=\frac 1 {\sqrt{2\pi}^d}\sum_{k,l=1}^d \int_{\mathbb{R}^d} -\frac 12 \partial_t\Sigma(t)_{k,l}\xi_k \xi_l e^{-\frac 12  \xi \cdot\Sigma(t) \xi}\hat g(\xi) \, d\xi
				\\
				&=\frac 1 {\sqrt{2\pi}^d}\sum_{k,l=1}^d \frac 12 \partial_t\Sigma(t)_{k,l} \bE \Bigl[\int_{\mathbb{R}^d}  -\xi_k \xi_l e^{-\mathrm{i} \xi \cdot W(t)}\hat g(\xi) \, d\xi \Bigr]
				\\
				&=\sum_{k,l=1}^d \frac12\partial_t\Sigma(t)_{k,l} \bE [ \partial_k \partial_l g(W(t))].
			\end{aligned}
		\end{gather*}
	\end{proof}	
	
	Next, we extend Lemma \ref{Lem_DerivSigma} to a wider class of functions $g$:
	\begin{lemma}
		\label{Lem_DerivSigma2}
		Let $d\in \mathbb{N}$, $\Sigma: [0,T]\rightarrow \mathbb{R}^{d \times d}$ a continuously differentiable, matrix valued map, such that $\Sigma(t)$ is positive semi-definite for all $t \in [0,T]$.
		
		Let $g \in C^2$ and suppose that $g$ has at most exponential growth in the sense of Definition \ref{Def_ExpGrowth}. 
		Then
		\begin{gather*}
			\partial_t \varphi(t) = \sum_{k,l=1}^d \frac 12 \partial_t \Sigma(t)_{k,l} \bE \bigl[\partial_k \partial_l g(W(t))\bigr]
		\end{gather*}
	\end{lemma}
	
	\begin{proof}
		Let $\lambda\in \mathbb{N}$ be such that
		\begin{gather*}
			\sup_{|x|>2} \frac{|g(x)|+|\partial g(x)|+|\partial^2g(x)|}{\exp(\lambda |x|)}\le 1. 
		\end{gather*}
		
		For $n \in \mathbb{N}$ let $g^{(n)}\in C^{\infty}$ be such that
		\begin{enumerate}
			\item $\operatorname{supp}(g^{(n)})\subset [-n-1,n+1]^d$,
			\item $\sup_{|x|\le n} |g(x)-g^{(n)}(x)|+|\partial g(x)-\partial g^{(n)}(x)|+
			|\partial^2 g(x)-\partial^2 g^{(n)}(x)|<1/n$,
			\item $|g^{(n)}(x)|+|\partial g^{(n)}(x)| + |\partial^2 g^{(n)}(x)|\le 4 \exp(\lambda |x|)$ for $|x| \ge n$. 
		\end{enumerate}
		Such a sequence can be constructed via $g^{(n)}=g\cdot\bigl(\mathds{1}_{[-n-1/2,n+1/2]}*\eta_{1/4}\bigr)$, where $\eta_{1/4}$ is a standard mollifier with $\operatorname{supp}\bigl(\eta_{1/4}\bigr)\subset [-1/4,1/4]$.
		Define $\varphi^{(n)}(t)=\bE[g^{(n)}(W(t))]$ .
		Note that $g^{(n)}$ satisfies the assumptions of Lemma \ref{Lem_DerivSigma}, implying that
		\begin{gather*}
			\varphi^{(n)}(t)-\varphi^{(n)}(0) =\int_0^t \sum_{k,l=1}^d \frac 12 \partial_s \Sigma(s)_{k,l} \bE \bigl[\partial_k \partial_l g^{(n)}(W(s))\bigr]\,ds.
		\end{gather*}
		We now show that $\varphi^{(n)}\rightarrow \varphi$ uniformly as $n \rightarrow \infty$. Let $\varepsilon>0$ be arbitrary. By (uniform) exponential growth, it follows from the generalized Hölder inequality, and the continuity of $\Sigma$ that there exists some $m \in \mathbb{N}$ such that 
		\begin{gather*}
			\sup_{t \in I} \sup_{n \in \mathbb{N}} \bE\bigl[g^{(n)}(W(t))\mathds{1}_{\{|X(t)|>m\}} \bigr] + \bE\bigl[g(W(t))\mathds{1}_{\{|W(t)|>m\}} \bigr]< \varepsilon.
		\end{gather*}
		W.l.o.g.\ we can assume that $1/m<\varepsilon$. By construction, we know that for every $n \ge m$ and all $t \in I$
		\begin{gather*}
			\bE\bigl[\bigl(g^{(n)}(W(t))-g(W(t))\bigr)\mathds{1}_{\{|W(t)|\le m\}} \bigr]  < \varepsilon.
		\end{gather*}
		Hence $\|\varphi-\varphi^{(n)}\|_\infty<2\varepsilon$ for $n \ge m$, implying uniform convergence. By using $\partial_k\partial_lg$ instead of $g$, we see that
		\begin{gather*}
			\lim_{n \rightarrow \infty} \bE \bigl[\partial_k \partial_l g^{(n)}(W(t))\bigr]=\bE \bigl[\partial_k \partial_l g(W(t))\bigr].
		\end{gather*}
		From dominated convergence it follows that 
		\begin{gather*}
			\varphi^{}(t)-\varphi^{}(0) =\int_0^t \sum_{k,l=1}^d \frac 12 \partial_s \Sigma(s)_{k,l} \bE \bigl[\partial_k \partial_l g^{}(W(s))\bigr]\,ds,
		\end{gather*}
		which proves the claim.
	\end{proof}
	
	\subsection{Multivariate estimates}
	
	In this section we demonstrate some estimates for the \emph{Liouville fractional Brownian motion} (recalling the Liouville Kernel equation \eqref{Liouville-Kernel})
	\begin{equation*}
		\whW_t = \int_0^t K(t, s)\,dW_s = \int_0^t (t-s)^{H-1/2} dW_s 
	\end{equation*}
	and its covariance function $C:\Delta_2^{\circ}\to \bR$, 
	\begin{equation}
		\label{eq:Covariance-Liouville}
		C(t, s) = \int_0^1 K(t, u)\cdot K(s, u)\,du = \int_0^{s\wedge t} (s - u)^{H-1/2} \cdot (t - u)^{H-1/2}\,du
	\end{equation}
	
	The following estimates are shown in \cite{Gassiat2022Weak}* {Proposition 4.1}: 
	\begin{align}
		C(t,t)&= C_H t^{2H}, \label{Lem_VarEstim1} 
		\\ 
		|C(t,t)-C(t,s)|+|C(t,s)-C(s,s)| &\lesssim (t-s)^{2H}, \label{Lem_VarEstim2} 
		\\ 
		|\partial_t C(t,s) |&\lesssim (t-s)^{2H-1}, \label{Lem_VarEstim3} 
		\\ 
		|\partial_s C(t,s) |&\lesssim  (t-s)^{2H-1}+s^{2H-1}. \label{Lem_VarEstim4}
	\end{align}
	For this work, we additionally require the following:
	\begin{lemma}
		\label{Lem_VarEstim}
		Let $t>s$. The covariance kernel of a Liouville fractional Brownian motion (Equation \eqref{eq:Covariance-Liouville}) satisfies the following estimates:
		\begin{align}
			|C(t,s)-C(t,\tilde s) |+|C(\tilde t,s)-C(t,s) |&\lesssim (\tilde s-s)^{2H}+ |\tilde t-t|^{2H} \label{Lem_VarEstim5}
		\end{align}
	\end{lemma}
	
	\begin{proof}
		Without loss of generality, we assume $\tilde s >s$ and write down explicitly
		\begin{align*}
			C(t,\tilde s)-C(t, s) &=\int_0^s K(t,r) \bigl(K(s,r)-K(\tilde s,r) \bigr) \, dr+ \int_s^{\tilde s \wedge t} K(t,r) K(\tilde s,r)  \, dr
			\\
			&\le
			\int_0^s K(s,r) \bigl(K(s,r)-K(\tilde s,r) \bigr) \, dr+ \int_0^{\tilde s - s} (\tilde s -s-r)^{2H-1}  \, dr
			\\
			&\le C( s,s)-C(s,\tilde s)+\frac 1 {2H} (\tilde s - s)^{2H} \lesssim (\tilde s - s)^{2H}.
		\end{align*}
		Assume that $\tilde t <t$. If $s>\tilde t$ we have
		\begin{align*}
			C(\tilde t, s)-C(t, s) &=\int_0^{\tilde t} K(s,r) \bigl(K(\tilde t,r)-K(t,r) \bigr) \, dr+ \int_{\tilde t}^{s} K(t,r) K(s,r)  \, dr
			\\
			&\le
			\int_0^{\tilde t} K(\tilde t,r) \bigl(K(\tilde t,r)-K(t,r) \bigr) \, dr+ \int_{\tilde t}^{s} K(s,r) K(s,r)  \, dr
			\\
			&\le C( t,\tilde t)-C(\tilde t,\tilde t)+\frac 1 {2H} (s - \tilde t)^{2H}
			\\
			&\lesssim (t  - \tilde t)^{2H}.
		\end{align*}
		If $s<\tilde t$ we use \eqref{Lem_VarEstim3} to see that
		\begin{align*}
			|C(\tilde t, s)-C(t, s)| &\lesssim \int_{\tilde t}^t (u-s)^{2H-1} du \lesssim \int_{\tilde t}^t (u-\tilde t)^{2H-1} du
			\lesssim (t-\tilde t)^{2H}.
		\end{align*}
	\end{proof}
	
	\begin{lemma}
		\label{Lem_VerifAss}
		Let $m \in \mathbb{N}$, $g \in C^{2}$ with at most exponential growth, see Definition \ref{Def_ExpGrowth}. We define $\varphi: \Delta_m^{\circ} \to \bR$ according to 
		\begin{equation*}
			\varphi(t_1,t_2,\dots,t_m) \coloneqq \bE\Bigl[ g(\whW_{t_1},\whW_{t_2},\dots,\whW_{t_m})\Bigr].
		\end{equation*}
		Then there exists a constant $C$ not depending on $t_1,\dots,t_m$ such that
		\begin{align}
			\label{eq:Lem_VerifAss_1}
			|\partial_j \varphi| & \le C\Bigl( (t_j-t_{j+1})^{2H-1}+(t_{j-1}-t_j)^{2H-1}\Bigr),
			\\\
			\label{eq:Lem_VerifAss_2}
			|\partial_1 \varphi| &\le C (t_1-t_2)^{2H-1},
			\\
			\label{eq:Lem_VerifAss_3}
			|\partial_m \varphi| & \lesssim C\Bigl( t_m^{2H-1}+(t_{m-1}-t_m)^{2H-1}\Bigr),
			\\
			\label{eq:Lem_VerifAss_4}
			\varphi(\mathbf{s})-\varphi\bigl(\mathbf{t}\bigr) &\lesssim  C\|\mathbf{s}-\mathbf{t}\|^{2H}.
		\end{align}
		In particular
        \begin{equation*}
            C \lesssim C'_g m \exp\bigg( \frac{(C_g \cdot m)^2}{2} \cdot \frac{T^{2H+1}}{(2H+1/2)^2} \bigg)
        \end{equation*}
	\end{lemma}
	
	In particular, $\varphi$ satisfies Assumption \ref{Ass_G} as will be discussed later.
	\begin{proof}
		Recall $\mathbf{t}=(t_1,\dots,t_m)\in \Delta_m^\circ$, $\partial_j \varphi=\partial_{t_j}\varphi$ and define the matrix $\Sigma(\mathbf{t})$ as
		$\Sigma_{k,l}(\mathbf{t})=C(t_k,t_l)=\bE\bigl[\whW_{t_k}\whW_{t_l}\bigr]$. Lemma \ref{Lem_DerivSigma2} implies that
		\begin{align*}
			\partial_j \varphi(\mathbf{t})= \sum_{k,l=1}^m \frac 12 \partial_j \Sigma(\mathbf{t})_{k,l} \bE \bigl[\partial_k \partial_l g(\whW_{t_1},\whW_{t_2},\dots,\whW_{t_m})\bigr]. 
		\end{align*}
		Courtesy of Equation \eqref{eq:Def_ExpGrowth}, we have that
		$$
		\sup_{k, l} \bE \Big[ \partial_k \partial_l g\big( \whW_{t_1},\whW_{t_2},\dots,\whW_{t_m} \big) \Big] 
		\leq 
		C'_g \cdot \exp\bigg( \frac{(C_g \cdot m)^2}{2} \cdot \frac{T^{2H+1}}{(2H+1/2)^2} \bigg) 
		\eqqcolon \fC.
		$$
		Using the explicit form of $\Sigma$ we get
		\begin{align*}
			|\partial_j \varphi|\lesssim \fC \bigg( \sum_{k=1}^{j-1} \Bigl|\partial_t C(t_k,t)\big|_{t=t_j}\Bigr|+ \Bigl|\partial_t C(t,t)\big|_{t=t_j}\Bigr| + \sum_{k=j+1}^{m} \Bigl|\partial_t C(t,t_k)\big|_{t=t_j}\Bigr| \bigg).
		\end{align*}
		Now Lemma \ref{Lem_VarEstim} implies for $1<j<m$ that
		\begin{gather}
			\label{Eq_VerifAss}
			\begin{aligned}
				|\partial_1 \varphi| &\lesssim m\fC  t_1^{2H-1}+m\fC(t_1-t_2)^{2H-1},
				\\
				|\partial_m \varphi| & \lesssim m\fC t_m^{2H-1}+m\fC (t_{m-1}-t_m)^{2H-1},
				\\
				|\partial_j \varphi| & \lesssim m\fC \Bigl( t_j^{2H-1}+(t_{j-1}-t_j)^{2H-1}+(t_j-t_{j+1})^{2H-1}\Bigr).
			\end{aligned}
		\end{gather}
		Here all the hidden constants do not depend on $T,g$ or $m$. Using the trivial bound $t_j^{2H-1}\lesssim(t_{j}-t_{j+1})^{2H-1}$ for $1\le j\le m-1 $, Equation \eqref{Eq_VerifAss} implies \eqref{eq:Lem_VerifAss_1}, \eqref{eq:Lem_VerifAss_2} and \eqref{eq:Lem_VerifAss_3}.
		
		Let $s_1> \dots>s_m$ and assume w.l.o.g. that $s_j \le t_j$ for $j=1,\dots, m$. Define
		\begin{gather*}
			\Delta_j \varphi = \varphi(t_1,\dots,t_{j-1},t_j,s_{j+1},\dots,s_m)-\varphi(t_1,\dots,t_{j-1},s_j,s_{j+1},\dots,s_m).
		\end{gather*}
		Then
		\begin{gather*}
			|\varphi(s_1,\dots,s_m)-\varphi(t_1,\dots,t_m)| \lesssim \sum_{j=1}^m |\Delta_j \varphi|.
		\end{gather*}
		Finally, by \eqref{Eq_VerifAss}, using concavity of $x^{2H}$, for all $1 \le j \le m$
		\begin{align*}
			|\Delta_j \varphi| \lesssim m\fC \int_{s_j}^{t_j} \Bigl(r^{2H-1} + (r-s_j)^{2H-1}+ (t_j-r)^{2H-1}\Bigr) \,dr \lesssim m\fC (t_j-s_j)^{2H},
		\end{align*}
		which implies \eqref{eq:Lem_VerifAss_4}.
	\end{proof}

	\section{Representation of Log-Stock}
	\label{Section:Representation}
	
	Recall the process $X$ is defined to be
	\begin{align*}
		X_t&= \rho \int_0^t f\bigl( \whW_s \bigr) dW_s + \sqrt{1-\rho^2} \int_0^t f\Bigl( \whW_s \Bigr) dW_s^\perp, 
		\\
		\whW_t &= \int_0^t (t-s)^{H-1/2} dW_s
		\quad
		W_t \perp W_t^{\perp}.
	\end{align*}
	The aim of this section is to derive an exact expression for $\bE [ (X_T)^N ]$. The case $N=2$ is elementary, It\^{o} isometry, we have that
	\begin{equation*}
		\bE \Bigl[ (X_T)^2 \Bigr] =\int_0^T \bE\Bigl[ f(\whW_t)^2 \Bigr] \,dt.
	\end{equation*}
	The case $N=3, \rho =1$ appeared in \cite{Gassiat2022Weak}, the general case left as open problem. There and below, the Clark-Ocone formula \cite[Prop 1.3.14]{nualart2006malliavin} is useful. It provides an explicit form of for It\^o's representation theorem, valid for sufficiently regular random variable on Wiener space, say $F \in D^{1,2}$ in Malliavin sense,
	\begin{equation*}
		F = \bE\bigl[ F \bigr] + \int_0^T \bE\bigl[ D_t F \big| \cF_t \bigr] dW_t
	\end{equation*}
	where $D$ denotes the Malliavin derivative. For the sequel it is sufficient to recall that, for deterministic $h \in L^2$ and $f \in C^1$ one has
	$D_t f ( \int_0^T h d W ) = f' (\int_0^T h d W ) h_t$. 
	
	\begin{example}
		\label{example:Gassiat-X^3}
		Following the ideas of \cite{Gassiat2022Weak}, let us consider the case $N=3$ as well as $\rho^2=1$. By redefining $f$ we can w.l.o.g.\ assume $\rho=1$. An application of It\^{o}'s formula yields that
		\begin{equation*}
			\bE\Bigl[ (X_T)^3 \Bigr] = 3\int_0^T \bE\Bigl[ X_t f\bigl( \whW_t \bigl)^2 \Bigr] dt. 
		\end{equation*}
		The key observation is that thanks to the Clark-Ocone formula 
		\begin{equation*}
			f\bigl( \whW_t \bigl)^2 = \bE\Bigl[ f\bigl( \whW_t \bigl)^2 \Bigr] + 2 \int_0^t \bE\Bigl[ ff'\bigl( \whW_t \bigl) \Big| \cF_s \Bigr] K(t, s) dB_s,
		\end{equation*}
		where $K(t, s)$ is the Liouville kernel defined in Equation \eqref{Liouville-Kernel}. Then
		\begin{equation*}
			3\int_0^T \bE\Bigl[ X_t f\bigl( \whW_t \bigl)^2 \Bigr] dt 
			= 
			6 \int_0^T \int_0^t \bE\Bigl[ f\bigl( \whW_s\bigl) ff'\bigl( \whW_t\bigl) \Bigr] K(t, s) ds dt.
		\end{equation*}
	\end{example}
	
	\begin{example}
		\label{example:Gassiat-X^4}
		Following the ideas of Example \ref{example:Gassiat-X^3}, we can also consider the case $N=4$. This was not addressed in \cite{Gassiat2022Weak}, but we provide it here for motivational purposes. Thanks to It\^{o}'s formula, 
		\begin{equation*}
			\bE\Bigl[ (X_T)^4 \Bigr] = 6 \int_0^T \bE\Bigl[ (X_t)^2 f \bigl( \whW_t\bigl)^2 \Bigr] dt. 
		\end{equation*}
		Similarly, thanks to It\^{o}'s formula and the Clark-Ocone formula we have that
		\begin{align*}
			X_t^2 =& 2 \int_0^t X_s f \bigl( \whW_s \bigl) dB_s + \int_0^t f\bigl( \whW_s\bigl)^2 ds, 
			\\
			f\bigl( \whW_t\bigl)^2 =& \bE\Bigl[ f\bigl( \whW_t\bigl)^2 \Bigr] + 2\int_0^t \bE\Bigl[ ff'\bigl( \whW_t\bigl) \Big| \cF_s \Bigr] K(t, s)\, dW_s
		\end{align*}
		where $K(t, s)$ is the Liouville kernel defined in Equation \eqref{Liouville-Kernel}. 
		
		Thus
		\begin{align*}
			\bE\Bigl[ (X_T)^4 \Bigr] &= 12 \int_0^T  \bE\Bigl[ f^2\bigl(\whW_{t}\bigr)\int_0^t X_s f\bigl( \whW_s\bigl)\,dB_s \Bigr] dt
			\\
			&\qquad+6 \int_0^T  \bE\Bigl[f\bigl( \whW_t\bigl)^2  \int_0^tf\bigl( \whW_s\bigl)^2 \,ds\Bigr] \,dt
			\\
			&=24 \rho \int_0^T \int_0^t \bE\Bigl[ X_s  f\bigl( \whW_s\bigl) f\bigl(\whW_{t}\bigr) f'\bigl(\whW_{t}\bigr)  \Bigr]\,ds \, dt 
			\\
			&\qquad+6 \int_0^T \int_0^t \bE\Bigl[ f\bigl( \whW_s\bigl)^2 f\bigl( \whW_t\bigl)^2 \Bigr] \, ds \, dt.
		\end{align*}
		Again, by taking the Malliavin derivative of the function $f\bigl( \whW_s\bigl) f f'\bigl( \whW_t \bigl)$ and applying the Clark-Ocone formula yields
		\begin{align*}
			\bE\Bigl[ (X_T)^4 \Bigr] 
			=&
			6 \int_0^T \int_0^t \bE\Bigl[ f\bigl( \whW_s\bigl)^2 f\bigl( \whW_t\bigl)^2 \Bigr] \, ds \, dt
			\\
			&+24\rho^2 \int_0^T \int_0^t \int_0^s \bE\Bigl[ f\bigl( \whW_r \bigl) f \bigl( \whW_s \bigl) ff''\bigl( \whW_t \bigl) \Bigr] K(t,s) K(t,r) \, dr\,ds\,dt
			\\
			&+24\rho^2 \int_0^T \int_0^t \int_0^s \bE\Bigl[ f\bigl( \whW_r \bigl) f \bigl( \whW_s \bigl) f'f'\bigl( \whW_t \bigl) \Bigr] K(t,s) K(t,r) \, dr\,ds\,dt
			\\
			&+24\rho^2 \int_0^T \int_0^t \int_0^s \bE\Bigl[ f\bigl( \whW_r \bigl) f'\bigl( \whW_s \bigl) ff'\bigl( \whW_t \bigl) \Bigr] K(t,s)K(s,r) \, dr\,ds\,dt.
		\end{align*}
	\end{example}
	
	\subsection{Polynomial test functions}
	
	We have seen in Example \ref{example:Gassiat-X^3} and Example \ref{example:Gassiat-X^4} that we can express fourth order moments of our rough volatility model in terms of the kernel $K$, the volatility functional $f$ and the correlation $\rho$. In this Section, we obtain a similar expression for higher order polynomials. 
	
	\begin{definition}
		\label{definition:I-J_operations}
		Let $m,N\in \bN$ and $\Delta_m^{\circ}$ be the open simplex as in Equation \eqref{eq:Simplex}. Let $\mathcal{Q}^{m}$ be the set of functions $F:\mathbb{R}^{m}\times \Delta_{m}^\circ \rightarrow \mathbb{R}$ such that for all $j\le m$ the derivative $\partial_{x_j}F(x_1,\dots,x_m,t_1,\dots,t_m)$ exists and is continuous. Furthermore set $\mathcal{Q}=\bigcup_{m \in \mathbb{N}} \mathcal{Q}^m$. 
		
		For $s<t_m$ and $y \in \mathbb{R}$, we define $\cI^N,\cJ^N:\mathcal{Q}^{m}\rightarrow C$  as
		\begin{gather*}
			(\cI^N F)(x_1,\dots,x_m,y,t_1,\dots,t_m,s)=\rho N f(y) \sum_{j=1}^m \partial_{x_j} F(x_1,\dots,x_m,t_1,\dots,t_m) K(t_j,s),
		\end{gather*}
		where $K(t, s)$ is the Liouville kernel defined in Equation \eqref{Liouville-Kernel} and
		\begin{gather*}
			(\cJ^N F)(x_1,\dots,x_m,y,t_1,\dots,t_m,s)=\frac{N(N-1)}2 f^2(y) F(x_1,\dots,x_m,t_1,\dots,t_m).
		\end{gather*}
	\end{definition}
	
	\begin{lemma}
		\label{Lem_IJ1}
		Let $N\geq 1$, $\cI,\cJ$ as before. Let $\mathbf{t}=(t_1, ..., t_{m}) \in \Delta_{m}^\circ$ and $t\in [0,t_{m})$. Let $F: \mathbb{R}^{m}\times \Delta_{m}^\circ \to \bR$ be in $\mathcal{Q}^m$ and denote $\whW_{\mathbf{t}}=\bigl( \whW_{t_1},\dots,\whW_{t_m}\bigr)$. Then
		\begin{align*}
			\bE\Bigl[ (X_t)^N F\bigl(\whW_{\mathbf{t}},\mathbf{t}\bigr) \Bigr]
			&=
			\int_0^t \bE\Bigl[ (X_s)^{N-1} \bigl(\cI^{N}(F)\bigr)\bigl(\whW_{\mathbf{t}},\whW_s,\mathbf{t},s\bigr) \Bigr] \,ds
			\\
			&+ 
			\int_0^t \bE\Bigl[ (X_s)^{N-2}  \bigl(\cJ^{N} (F)\bigr)\bigl(\whW_{\mathbf{t}},\whW_s,\mathbf{t},s\bigr)\Bigr] \,ds.
		\end{align*}
	\end{lemma}
	
	\begin{remark} 
		Since $\cJ^N(F)\equiv 0$ for $N= 1$ we define $(X_s)^{N-2}\cJ^N(F)= 0$, meaning that the second term above vanishes in this case.
	\end{remark}
	
	\begin{proof}
		We apply It\^{o} formula for the function $x \mapsto x^N$ to get
		\begin{equation*}
			(X_t)^N = N \int_0^t (X_s)^{N-1} f\bigl( \whW_s \bigl) \, dB_s + \frac{N(N-1)}2 \int_0^t (X_s)^{N-2} f\bigl( \whW_s \bigl)^2 \,ds.    
		\end{equation*}
		Define $\partial_j F=\partial_{x_j}F.$
		Using the definition of $F$, we apply the Clark-Ocone formula to see that
		\begin{align*}
			F\bigl(\whW_{\mathbf{t}},\mathbf{t}\bigr)&=\bE\Bigl[ F\bigl(\whW_{\mathbf{t}},\mathbf{t}\bigr) \Bigr] 
			+ 
			\sum_{j=1}^{m} \int_0^{t_j} \bE\Bigl[\partial_j F\bigl(\whW_{\mathbf{t}},\mathbf{t} \bigr) \Big| \cF_s \Bigr]K(t_j,s) \, dW_s.
		\end{align*}
		Using both representations we get
		\begin{align*}
			\bE\Bigl[ (X_t)^NF\bigl(\whW_{\mathbf{t}},\mathbf{t}\bigr) \Bigr]
			=&\bE\Bigl[ N\Bigl(\int_0^t (X_s)^{N-1} f\bigl( \whW_s \bigl) \, dB_s\Bigr) F\bigl(\whW_{\mathbf{t}},\mathbf{t} \bigr) \Bigr]
			\\
			&+
			\bE \Bigl[ \frac{N(N-1)}2 \Bigl(\int_0^t (X_s)^{N-2} f^2\bigl( \whW_s \bigl) \,ds\Bigr) F\bigl(\whW_{\mathbf{t}},\mathbf{t} \bigr) \Bigr]
			\\
			=& \int_0^t \bE\Bigl[ \sum_{j=1}^m \rho N (X_s)^{N-1} f\bigl( \whW_s \bigl) \partial_j F\bigl(\whW_{\mathbf{t}},\mathbf{t} \bigr) \Bigr] K(t_j,s)\, ds
			\\
			&+ 
			\int_0^t \bE \Bigl[ (X_s)^{N-2} \frac{N(N-1)}2 f^2(\whW_s)F\bigl(\whW_{\mathbf{t}},\mathbf{t}\bigr) \Bigr] \,ds 
			\\
			=&\int_0^t \bE\Bigl[ (X_s)^{N-1} \bigl(\cI^{N}(F)\bigr)\bigl(\whW_{\mathbf{t}},\whW_s,\mathbf{t},s\bigr) \Bigr] \,ds
			\\
			&+ 
			\int_0^t \bE\Bigl[ (X_s)^{N-2}  \bigl(\cJ^{N} (F)\bigr)\bigl(\whW_{\mathbf{t}},\whW_s,\mathbf{t},s\bigr)\Bigr] \,ds.
		\end{align*}
	\end{proof}
	
	\begin{definition}
		Let $\cW$ be the set of all words with letters in $\{\wI,\wJ\}$. We denote by $|.|$ the number of letters, i.e.\ for $w=w_1\dots w_m$ we have $|w|=m$.
		
		We define an inhomogenous length $\ell: \cW\to \bN$ via  $\ell(\wI)=1$, $\ell(\wJ)=2$ and 
		\begin{equation*}
			\ell(w)=\sum_{j=1}^{|w|} \ell(w_j). 
		\end{equation*}
		We define an embedding $\iota: \cW \to \mathcal{Q}$ as follows: If $w$ has a single letter $w=\wI$ or $w=\wJ$ then 
		\begin{equation*}
			\iota(\wI)=\cI^1
			\quad\mbox{or}\quad
			\iota(\wJ)=\cJ^2.
		\end{equation*}
		For $w=w_1\dots w_m$ we define
		\begin{equation*}
			\iota(w)=
			\begin{cases}
				\iota(w_1\dots w_{m-1}) \circ \wI^{\ell(w)} \quad &\quad \mbox{if} \quad w_m = I
				\\
				\iota(w_1\dots w_{m-1}) \circ \cJ^{\ell(w)} \quad &\quad \mbox{if} \quad w_m = J. 
			\end{cases}
		\end{equation*}
	\end{definition}
	
	\begin{lemma}
		\label{Lem_Izero}
		Let $w \in \cW$ ba a word with last letter $w_{|w|}=\cI$. Denote by $1$ the constant one function. Then $\iota(w)(1)\equiv 0$.
	\end{lemma}
	\begin{proof}
		By definition of $\iota$ it holds that $\iota(w)(1)=\iota(w_1\dots w_{|w|-1})(\cI^{\ell(w)} 1)$. By definiton of the operator $\cI^{\ell(w)}$ it follows that $\cI^{\ell(w)} 1\equiv0$.
	\end{proof}
	
	\begin{proposition}
		\label{Prop_MomRep}
		Let $N\geq 1$, $\cI,\cJ$ as before. Let $t<u_m<\dots<u_1$ and assume furthermore that $F(x_1,\dots,x_{m},u_1,\dots,u_m)$ is $N$-times continuously differentiable in its first $m$ variables as well as that $f\in C^{N-2}$.  Then
		\begin{align*}
			\bE&\Bigl[ (X_t)^N F\bigl(\whW_{\mathbf{u}},\mathbf{u}\bigr) \Bigr]
			\\
			&=
			\sum_{\substack{w\in \cW \\ \ell(w)=N}}\,\, \int_0^t ... \int_0^{t_{|w|-1}} \bE\Bigl[\bigl( \iota(w)F\bigr)\bigl(\whW_{\mathbf{u}},\whW_{t_1},\dots,\whW_{t_{|w|}},\mathbf{u},t_1,\dots,t_{|w|}\bigr) \Bigr] dt_{|w|} \dots dt_1.
		\end{align*}
	\end{proposition}
	
	\begin{proof}
		We prove this lemma by induction. For $N=1$ this is already shown in Lemma \ref{Lem_IJ1}. 
		
		Recall $\mathbf{u}=(u_1,\dots,u_m)$ as well as $\whW_{\mathbf{u}}=(\whW_{u_1},\dots,\whW_{u_m})$.
		For the induction step we also use Lemma \ref{Lem_IJ1}, implying that
		\begin{align*}
			\bE\Bigl[ (X_t)^N  F\bigl(\whW_{\mathbf{u}},\mathbf{u}\bigr) \Bigr]
			=&\int_0^t \bE\Bigl[ (X_{t_1})^{N-1}\bigl(\cI^N F \bigr)\bigl(\whW_{\mathbf{u}},\whW_{t_1},\mathbf{u},t_1\bigr) \Bigr] dt_1
			\\
			&+
			\int_0^t \bE\Bigl[ (X_{t_1})^{N-2} \bigl(\cJ^NF\bigr)\bigl(\whW_{\mathbf{u}},\whW_{t_1},\mathbf{u},t_1\bigr) \Bigr]
			dt_1.
		\end{align*}
		\pagebreak[1]
		
		By the induction hypothesis
		\begin{align*}
			\bE\Bigl[& (X_t)^N  F\bigl(\whW_{\mathbf{u}},\mathbf{u}\bigr) \Bigr]
			\\
			=&\mkern-18mu \sum_{\substack{ v \in \cW\\ \ell(v)=N-1}} \, \int\limits_{t\ge t_1\ge \dots\ge t_{|v|+1}}
			\mkern-28mu 
			\bE\bigl[ \bigl(\iota(v)(\cI^N F)\bigr)\bigl(\whW_{\mathbf{u}},\whW_{t_1},\dots,\whW_{t_{|v|+1}},\mathbf{u},t_{1},\dots,t_{|v|+1}\bigr) \bigr]
			\,dt_{|v|+1}\dots dt_1
			\\
			+&\mkern-18mu \sum_{\substack{ v \in \cW\\ \ell(v)=N-2}} \, \int\limits_{t\ge t_1\ge \dots\ge t_{|v|+1}}
			\mkern-28mu 
			\bE\bigl[ \bigl(\iota(v)(\cJ^N F)\bigr)\bigl(\whW_{\mathbf{u}},\whW_{t_1},\dots,\whW_{t_{|v|+1}},\mathbf{u},t_{1},\dots,t_{|v|+1}\bigr) \bigr]
			\,dt_{|v|+1}\dots dt_1
			\\
			=&\sum_{\substack{w\in \cW: \ell(w)=N\\ w_{|w|}=\wI}} \,  \int\limits_{t\ge t_1\ge \dots\ge t_{|w|}}
			\bE\bigl[ \bigl(\iota(w)F\bigr)\bigl(\whW_{\mathbf{u}},\whW_{t_1},\dots,\whW_{t_{|w|}},\mathbf{u},t_{1},\dots,t_{|w|}\bigr) \bigr]
			\,dt_{|w|}\dots dt_1
			\\
			&+\sum_{\substack{w\in \cW: \ell(w)=N\\ w_{|w|}=\wJ}} \, \int\limits_{t\ge t_1\ge  \dots\ge t_{|w|}}\bE\bigl[ \bigl(\iota(w)F\bigr)\bigl(\whW_{\mathbf{u}},\whW_{t_1},\dots,\whW_{t_{|w|}},\mathbf{u},t_{1},\dots,t_{|w|}\bigr) \bigr]
			\,dt_{|w|}\dots dt_1
			\\
			=&\sum_{\substack{w\in \cW \\ \ell(w)=N}}\,\, \int_0^t ... \int_0^{t_{|w|-1}} \bE\Bigl[\bigl( \iota(w)F\bigr)\bigl(\whW_{\mathbf{u}},\whW_{t_1},\dots,\whW_{t_{|w|}},\mathbf{u},t_1,\dots,t_{|w|}\bigr) \Bigr] dt_{|w|} \dots dt_1.
		\end{align*}
		Here we used the fact that for $w$ with $\ell(w)=N$ we have either 
		\begin{equation*}
			\iota(w)F=\iota(w_1\dots w_{|w|-1})\cI^N F
			\quad \mbox{or} \quad 
			\iota(w)F=\iota(w_1\dots w_{|w|-1})\cJ^NF.
		\end{equation*}
	\end{proof}
	
	\begin{theorem}
		\label{Thm_MomRep}
		Let $N\ge 1$. Then
		\begin{equation}
			\label{eq:Thm_MomRep}
			\begin{aligned}
				&\bE\Bigl[ (X_T)^N \Bigr]
				\\
				&=
				\sum_{\substack{w\in \cW \\ \ell(w)=N}}\,\, \int_0^T \int_0^{t_1} \dots \int_0^{t_{|w|-1}} \bE\Bigl[\bigl( \iota(w)1\bigr)\bigl(\whW_{t_1},\dots,\whW_{t_{|w|}},t_1,\dots,t_{|w|}\bigr) \Bigr] dt_{|w|} \dots dt_1.
			\end{aligned}
		\end{equation}
	\end{theorem}
	
	\begin{proof}
		This theorem follows immediately from Proposition \ref{Prop_MomRep} with $F\equiv 1$.
	\end{proof}
	
	\begin{example}
		In the case $N=3$ there are $3$ words $w$ with $\ell(w)=3$, namely $\wI\wI\wI$, $\wJ\wI$ and $\wI\wJ$. By Lemma \ref{Lem_Izero} we have $\iota(\wI)1=0$ and therefore the only non-trivial word is $\wI\wJ$. It follows that
		\begin{align*}
			\bigl(\iota(\wI\wJ)1\bigr)(y,x,t,s)&=\bigl(\cI^1 \cJ^3 1\bigr)(y,x,t,s)=3 \bigl( \cI^1 f^2\bigr)(y,x,t,s)
			=6\rho f\bigl( x \bigr) ff'\bigl( y \bigr)K(t,s).
		\end{align*}
		Thus it follows that
		\begin{align*}
			\bE\Bigl[ (X_T)^3 \Bigr]=6\rho \int_0^T \int_0^t \bE\Bigl[ f\bigl( \whW_s \bigr) ff'\bigl( \whW_t \bigr)\Bigr] K(t,s) \,ds\,dt.
		\end{align*}
		This expression agrees with Example \ref{example:Gassiat-X^3}.
	\end{example}
	
	\begin{example}
		\label{Ex_MomRep2}
		For $N=4$ the non-trivial words are $\wJ\wJ$ and  $\wI\wI\wJ$. A short calculation shows that
		\begin{align*}
			\bigl(\iota(\wJ\wJ)1\bigr)(x,y,t,s)=\bigl(\cJ^2\cJ^41\bigr)(x,y,t,s)=6f^2(x)f^2(y).
		\end{align*}
		Furthermore we see that $\bigl(\cI^2\cJ^41\bigr)(z,y,t,s)=12\rho f(y)\cdot 2 ff'(z) K(t,s)$.
		Applying $\mathcal{I}^1$ to this function we get
		\begin{align*}
			\bigl(\cI^1(\cI^2\cJ^41)\bigr)&(z,y,x,t,s,r)
			\\
			&=24\rho^2 f(x)\Bigl( f'(y) ff'(z)K(s,r)
			+f(y)\partial_z \bigl(ff'(z)\bigr)K(t,r)\Bigr)K(t,s)
			\\
			&=24\rho^2 f(x)f'(y)ff'(z)K(s,r)K(t,s)
			\\
			&\quad+
			24\rho^2 f(x)f(y)f'f'(z)K(t,r)K(t,s)
			\\
			&\quad +24\rho^2 f(x)f(y)ff''(z)K(t,r)K(t,s).
		\end{align*}
		Using Theorem \ref{Thm_MomRep} and the fact that $\iota(\wI\wI\wJ)1=(\cI^1(\cI^2\cJ^41)$ it follows that
		\begin{align*}
			\bE\Bigl[ (X_T)^4 \Bigr]=&6 \int_0^T \int_0^t \bE\Bigl[ f\bigl( \whW_s \bigr)^2 f\bigl( \whW_t \bigr)^2 \Bigr] \,ds\,dt
			\\
			&+24\rho^2 \int_0^T \int_0^t \int_0^s \bE\Bigl[ f\bigl( \whW_r \bigr) f'\bigl( \whW_s \bigr) f\bigl( \whW_t \bigr) f'\bigl( \whW_t \bigr) \Bigr] K(t,s)K(s,r) \, dr\,ds\,dt,
			\\
			&+24\rho^2 \int_0^T \int_0^t \int_0^s \bE\Bigl[ f(\whW_r) f(\whW_s)  f'\bigl( \whW_t \bigr) f'\bigl( \whW_t \bigr) \Bigr] K(t,s)K(t,r) \, dr\,ds\,dt
			\\
			&+24\rho^2 \int_0^T \int_0^t \int_0^s \bE\Bigl[ f(\whW_r) f(\whW_s)  f\bigl( \whW_t \bigr) f''\bigl( \whW_t \bigr) \Bigr] K(t,s)K(t,r) \, dr\,ds\,dt,
		\end{align*}
		which agrees with Example \ref{example:Gassiat-X^4}. 
	\end{example}
	
	\subsection{Representation of the integrand}
	
	Let $m \in \mathbb{N}$ and consider a word $w\in \cW$ such that $|w|=m$. Let $1$ denote the constant function with value $1$.
	Our goal is to find a compact representation for the expression $\bigl( \iota(w) 1 \bigr)(x_1,\dots,x_m,t_1,\dots,t_m)$. 
	\begin{definition}
		\label{Def_Afcts}
		For $w \in \cW$ such that $|w|=m$, let 
		\begin{equation*}
			N_J^w=\{j: w_j = \wJ\} \quad \mbox{and} \quad N_I^w=\{1,\dots,m\}\setminus N_J^w. 
		\end{equation*}
		Let $k\coloneqq |N_I^w|\ge 1$ and use an enumeration $N_I^w=\{j_1,\dots,j_k\}$ such that $j_1<\dots<j_k$. 
		
		We define 
		\begin{equation*}
			\cL^w= \Big\{ \mathbf{l}=(l_1,\dots,l_k) \in \bN^{\times k}: l_1\le m-j_k,\dots,l_k\le m-j_1 \Big\}.
		\end{equation*}
		For $\mathbf{l}\in \cL^w$, we define $\alpha_\mathbf{l}: \{1,\dots,m\} \rightarrow \{0,1,\dots,m\}$ such that 
		\begin{align*}
			&\alpha_\mathbf{l}(m-j_{k-i}+1)=l_{i} 
			\quad \text{for } i=1,\dots,k-1,
			\\
			&\alpha_\mathbf{l}(j)=0 \quad \forall j \in \{m-j+1:j \in N_J^w\}. 
		\end{align*}
		If $k=0$ we define $\cL^w=\{1\}$ and set $\alpha_1\equiv 0$.
	\end{definition}
	
	Recall that $w_m=\wJ$, as otherwise $w1\equiv 0$ by Lemma \ref{Lem_Izero}.
	\pagebreak[1]
	\begin{proposition}
		\label{Lem_wRep}
		Fix $m\in \bN$ and $w\in \cW$ such that $|w|=m$. For $\mathbf{l} \in \cL^w$, we define
		$\psi_{\mathbf{l}}: \bR^{\times m} \to \bR$ to be 
		\begin{equation}
			\label{eq:Lem_wRep}
			\psi_{\mathbf{l}}(x_1,\dots,x_m)= \partial_{x_{l_1}}\dots \partial_{x_{l_{k}}} 
			\prod_{\substack{l \in N_I^w}}f(x_{m-l+1}) \prod_{\substack{l \in N_J^w}}f^2(x_{m-l+1}).
		\end{equation}
		Then there exist a constant $C_w$ only depending on $w$ and $\rho$ such that
		\begin{align}
			\nonumber
			(\iota(w)1)&(x_1,\dots,x_m,t_1,\dots,t_m)
			\\
			\nonumber
			=&C_w\sum_{l_1=1}^{m-j_k} \dots \sum_{l_{k}=1}^{m-j_{1}} \partial_{x_{l_1}}\dots \partial_{x_{l_{k}}} \biggl( \prod_{\substack{l \in N_J^w}}f^2(x_{m-l+1}) \prod_{\substack{l \in N_I^w}}f(x_{m-l+1}) \biggr)
			\\
			\nonumber
			&\mkern 300mu \cdot K(t_{l_1},t_{m-j_k+1}) \dots K(t_{l_{k}},t_{m-j_{1}+1})
			\\
			\label{Eq_wRep}
			=&C_w \sum_{\mathbf{l} \in \cL^w} \psi_{\mathbf{l}}(x_1,\dots,x_m)K(t_{\alpha_{\mathbf{l}}(2)},t_2) \dots K(t_{\alpha_{\mathbf{l}}(m)},t_m).
		\end{align}
	\end{proposition}
	
	\begin{proof}
		We prove this claim by induction over $k$. We denote by $C$ a constant which does not depend on $t_1,\dots,t_m$ or $x_1,\dots,x_m$ and which may change every line.
		
		If $k=1$ then $w$ can be written as $w=w_1 \wI w_2$ where $w_1=\wJ^{j-1}$ and $w_2=\wJ^{m-j}$. By construction of $\iota$ it follows that $\iota(w)1=C \bigl(\iota(w_1)\bigl(\iota(\wI) (\iota(w_2)1)\bigr)\bigr)$. From Definition \ref{definition:I-J_operations} it follows that $\bigl(\iota(w_2)1\bigr)(x_1,\dots,x_{m-j})=C \prod_{l=1}^{m-j}f^2(x_l)$.
		Now we apply $\iota(\wI)$ and get
		\begin{gather*}
			\Bigl(\iota(\wI) \bigl(\iota(w_2)1\bigr) \Bigr)(x_1,\dots,x_{m-j+1})=C\sum_{i=1}^{m-j} f(x_{m-j-1})K(t_i,t_{m-j+1})\partial_{x_i} \prod_{l=1}^{m-j}f^2(x_l) .
		\end{gather*}
		Applying $\iota(w_1)$ to this function it follows that
		\begin{align*}
			(\iota(w)1)&=C\prod_{l=m-j+2}^{m}f^2(x_l) \sum_{i=1}^{m-j}  K(t_i,t_{m-j+1})f(x_{m-j-1})\partial_{x_i}\prod_{l=1}^{m-j}f^2(x_l)
			\\
			&=C \sum_{i=1}^{m-j}  K(t_i,t_{m-j+1}) f(x_{m-j-1}) \partial_{x_i}\prod_{l \in N_J^w}f^2(x_{m-l+1}).
		\end{align*}
		
		Now let $k>1$ be arbitrary. We split $w$ up into $w=w_1\wI w_2$ where $w_1=\wJ^{j_1-1}$ and such that the letter $\wI$ appears in $w_2$ $k-1$ times.
		
		By induction hypothesis we have
		\begin{align*}
			\bigl(\iota(w_2)&1\bigr)(x_1,\dots,x_{m-j_1},t_1,\dots,t_{m-j_1})  
			\\
			&=C\sum_{l_1=1}^{m-j_k} \dots \sum_{l_{k-1}=1}^{m-j_{2}}\partial_{x_{l_1}}\dots \partial_{x_{l_{k-1}}} 
			\prod_{\substack{l \in N_I^w\\ l>j_1}}f(x_{m-l+1}) \prod_{\substack{l \in N_J^w\\ l>j_1}}f^2(x_{m-l+1}) 
			\\
			&\mkern 300mu \cdot K(t_{l_1},t_{m-j_k+1}) \dots K(t_{l_{k-1}},t_{m-j_{2}+1}).
		\end{align*}
		By definition of $\iota(\wI)$ we have
		\begin{align*}
			\Bigl(\iota(\wI)&\bigl(\iota(w_2)1\bigr)\Bigr)(x_1,\dots,x_{m-j_1+1},t_1,\dots,t_{m-j_1+1})  
			\\
			&=C\sum_{l_k=1}^{m-j_1} \partial_{x_{l_k}} \bigl(\iota(w_2)1\bigr)(x_1,\dots,x_{m-j_1},t_1,\dots,t_{m-j_1})f(x_{m-j_1+1}) K(t_{l_k},t_{m-j_1+1})
			\\
			&=C\sum_{l_1=1}^{m-j_k} \dots \sum_{l_{k}=1}^{m-j_{1}}\partial_{x_{l_1}}\dots \partial_{x_{l_{k}}} 
			\prod_{\substack{l \in N_I^w}}f(x_{m-l+1}) \prod_{\substack{l \in N_J^w\\ l>j_1}}f^2(x_{m-l+1})
			\\
			&\mkern 300mu \cdot K(t_{l_1},t_{m-j_k+1}) \dots K(t_{l_{k-1}},t_{m-j_{1}+1}).
		\end{align*}
		
		From Definition \ref{definition:I-J_operations} it follows that
		\begin{align*}
			\Bigl(\iota(&w_1)\bigl(\iota(\wI w_2)1\bigr)\Bigr)(x_1,\dots,x_{m},t_1,\dots,t_{m})
			\\
			&=C\sum_{l=1}^{j_1-1} f^2(x_{m-l})\bigl(\iota(\wJ w_2)1\bigr)(x_1,\dots,x_{m-j_1+1},t_1,\dots,t_{m-j_1+1})  
			\\
			&=C\sum_{l_1=1}^{m-j_k} \dots \sum_{l_{k}=1}^{m-j_{1}}\partial_{x_{l_1}}\dots \partial_{x_{l_{k}}} 
			\prod_{\substack{l \in N_I^w}}f(x_{m-l+1}) \prod_{\substack{l \in N_J^w}}f^2(x_{m-l+1})
			\\
			&\mkern 300mu \cdot K(t_{l_1},t_{m-j_k+1}) \dots K(t_{l_{k}},t_{m-j_{1}+1}).
		\end{align*}
		
		Now it follows from a simple substitution that
		\begin{align*}
			(\iota(w)1)&(x_1,\dots,x_m,t_1,\dots,t_m)=\sum_{\mathbf{l}\in \cL^w} \psi_{\mathbf{l}}(x_1,\dots,x_m) K(t_{\alpha_{\mathbf{l}}(2)},t_2)\dots K(t_{\alpha_{\mathbf{l}}(m)},t_m).
		\end{align*}
	\end{proof}
	
	\begin{remark}\label{Rem_Cw}
		Let $w \in \cW$ with  $w_{|w|}=\wJ$. Following Definition \ref{definition:I-J_operations} and the proof of Proposition \ref{Lem_wRep}, manual computation allows us to verify that the constant $C_w$ satisfies
		\begin{equation*}
			|C_w|=|\rho|^{2|w|-\ell(w)}2^{\ell(w)-|w|}\ell(w)!.
		\end{equation*}
	\end{remark}
	
	\begin{example}
		In Example \ref{Ex_MomRep2} we looked in the fourth moment at the word $\wI\wI\wJ$. Written out we had 
		\begin{align*}
			\bigl(\iota(\wI\wI\wJ)1)\bigr)(x_1,x_2,x_3,t_1,t_2,t_3)
			&=24\rho^2 f(x_3)f(x_2)\partial_{x_1}\bigl(f(x_1)f'(x_1)\bigr)K(t_1,t_2)K(t_1,t_3)
			\\
			&\quad +24 \rho^2 f(x_3)f'(x_2)f(x_1)f'(x_1)K(t_1,t_2)K(t_2,t_3).
		\end{align*}
		In the notation of this section we would have $m=3$, $N_J^w=\{3\}$, $N_I^w=\{1,2\}$ as well as $\cL^w=\{(1,1),(1,2) \}$. Now define the functions $\alpha_{1,1},\alpha_{1,2}:\{1,2,3\}\rightarrow \{0,1,2,3\}$ as well as $\psi_{{1,2}},\psi_{{1,1}}:\mathbb{R}^3\rightarrow\mathbb{R}$ such that
		\begin{enumerate}
			\item $\alpha_{1,1}(1)=0$, $\alpha_{1,1}(2)=1$, $\alpha_{1,1}(3)=1$,
			\item $\alpha_{1,2}(1)=0$, $\alpha_{1,2}(2)=1$, $\alpha_{1,2}(3)=2$,
			\item $\psi_{{1,1}}(x_1,x_2,x_3)=\partial_{x_1}\partial_{x_1} f(x_3)f(x_2)f^2(x_1)=2 f(x_3)f(x_2)\partial_{x_1}\bigl(f(x_1)f'(x_1)\bigr)$,
			\item $\psi_{{1,2}}(x_1,x_2,x_3)=\partial_{x_1}\partial_{x_2} f(x_3)f(x_2)f^2(x_1)=2 f(x_3)f'(x_2)f(x_1)f'(x_1)$,
			\item $C_{\wI\wI\wJ}=12\rho^2$.
		\end{enumerate}
		Then it follows by putting everything together
		\begin{align*}
			\bigl(\iota(\wI\wI\wJ)1)\bigr)(x_1,x_2,x_3,t_1,t_2,t_3)
			&= C_{\wI\wI\wJ}\psi_{{1,1}}(x_1,x_2,x_3)K(t_{\alpha_{1,1}(2)},t_2)K(t_{\alpha_{1,1}(3)},t_3)
			\\
			&\quad + C_{\wI\wI\wJ}\psi_{{1,2}}(x_1,x_2,x_3)K(t_{\alpha_{1,2}(2)},t_2)K(t_{\alpha_{1,2}(3)},t_3).
		\end{align*}
		
		The other word appearing was $w=\wJ \wJ$.
		Written out we had 
		\begin{align*}
			\bigl(\iota(\wJ\wJ)1)\bigr)(x_1,x_2,t_1,t_2)
			&=6 f^2(x_1)f^2(x_2).
		\end{align*}
		In the notation of this section we would have $m=2$, $N_J^w=\{1,2\}$ as well as $N_I^w=\emptyset$. By Definition \ref{Def_Afcts} we set $\cL^w=\{1\}$. Therefore we have $\alpha_1(1)=\alpha_2(1)=0$ as well as $\psi_1(x_1,x_2)=f^2(x_1)f^2(x_2)$. The constant $C_w$ is given by $C_w=6$.
		Then it follows by putting everything together
		\begin{align*}
			\bigl(\iota(\wJ\wJ)1)\bigr)(x_1,x_2,t_1,t_2)
			&=6 f^2(x_1)f^2(x_2)=C_w \psi_1(x_1,x_2)K(t_{\alpha_{1,2}(2)},t_2).
		\end{align*}
	\end{example}
	
	\begin{example}
		Consider the word $w=\wI \wJ\wJ$. By Definition \ref{definition:I-J_operations} we have
		\begin{align*}
			\iota(w)1&=\bigl(\wI^1\wJ^3\wJ^51\bigr)(x_1,x_2,x_3,t_1,t_2,t_3)
			\\
			&=30\rho \partial_{x_1} f(x_3) f^2(x_2)f^2(x_1) K(t_1,t_3)+30\rho \partial_{x_2} f(x_3) f^2(x_2)f^2(x_1) K(t_1,t_2).
		\end{align*}
		For this word we would have $m=3$, $N_J^w=\{2,3\}$ $N_I^w=\{1\}$ as well as $\cL^w=\{1,2 \}$. 
		Define the functions $\alpha_{1},\alpha_{2}:\{1,2,3\}\rightarrow \{0,1,2,3\}$ as well as $\psi_{{1}},\psi_{{2}}:\mathbb{R}^3\rightarrow\mathbb{R}$ such that
		\begin{enumerate}
			\item $\alpha_{1}(1)=0$, $\alpha_{1}(2)=0$, $\alpha_{1}(3)=1$,
			\item $\alpha_{2}(1)=0$, $\alpha_{2}(2)=0$, $\alpha_{2}(3)=2$,
			\item $\psi_{{1}}(x_1,x_2,x_3)= \partial_{x_1} f(x_3) f^2(x_2)f^2(x_1)$,
			\item $\psi_{{2}}(x_1,x_2,x_3)= \partial_{x_2} f(x_3) f^2(x_2)f^2(x_1)$,
			\item $C_w=30\rho$.
		\end{enumerate}
		Then we have
		\begin{align*}
			\bigl(\iota(w)1)\bigr)(x_1,x_2,x_3,t_1,t_2,t_3)
			&= C_w\psi_{{1}}(x_1,x_2,x_3)K(t_{\alpha_{1}(2)},t_2)K(t_{\alpha_{1}(3)},t_3)
			\\&\qquad+ C_w\psi_{{2}}(x_1,x_2,x_3)K(t_{\alpha_{2}(2)},t_2)K(t_{\alpha_{2}(3)},t_3)
			\\
			&=C_w\psi_{{1}}(x_1,x_2,x_3)K(t_{\alpha_{1}(3)},t_3)
			\\
			&\qquad+ C_w\psi_{{2}}(x_1,x_2,x_3)K(t_{\alpha_{2}(3)},t_3).
		\end{align*}
	\end{example}
	
	\subsection{Discrete time approximations}
	
	Next, we demonstrate that the stochastic process $X_T^n$ defined in Equation \eqref{eq:RoughVolatilityModel-Discrete} has a similar decomposition as Theorem \ref{Thm_MomRep} (see Theorem \ref{Thm_MomRepDisr} below). To do this, we describe operators (see \ref{Def_DiscWords} below) that have the same It\^{o}'s formula type properties to those of Definition \ref{definition:I-J_operations} but for the process $X_T^n$ and then iterate these to get a formula equivalent to Equation \eqref{eq:Thm_MomRep}. 
	\begin{definition}
		\label{Def_DiscWords}
		Fix $n \in \mathbb{N}$ and recall $\eta(t)=\lfloor nt\rfloor/n$. For $K$ the Liouville kernel defined in Equation \eqref{Liouville-Kernel}, we denote
		\begin{equation}
			\label{eq:Def_DiscWords_K}
			\tilde K(t,s)=K(\eta(t),s).
		\end{equation}
		For $k \in \mathbb{N}$, we define $\tilde{\cI}^N: \mathcal{Q}^{m} \to C$ as
		\begin{equation*}
			(\tilde{\cI}^N F)(x_1,\dots,x_m,y,t_1,\dots,t_m,s)=\rho N f(y) \sum_{j=1}^m \partial_{x_j} F(x_1,\dots,x_m,t_1,\dots,t_m) \tilde{K}(t_j,s). 
		\end{equation*}
		
		We define an embedding $\tilde{\iota}: \cW \to \mathcal{Q}$ as follows: If $w$ has a single letter then 
		\begin{equation*}
			\tilde{\iota}(\wI)= \tilde{\cI}^1
			\quad\mbox{or}\quad
			\tilde{\iota}(\wJ)=\cJ^2.
		\end{equation*}
		For $w=w_1\dots w_m$ we define
		\begin{equation*}
			\tilde{\iota}(w)=
			\begin{cases}
				\tilde{\iota}(w_1\dots w_{m-1}) \circ \tilde{\cI}^{\ell(w)} \quad &\quad \mbox{if} \quad w_m = I
				\\
				\tilde{\iota}(w_1\dots w_{m-1}) \circ \cJ^{\ell(w)} \quad &\quad \mbox{if} \quad w_m = J. 
			\end{cases}
		\end{equation*}
	\end{definition}
	
	\begin{theorem}
		\label{Thm_MomRepDisr}
		Let $N\ge 1$. Then
		\begin{align*}
			\bE&\Bigl[ (X_T^n)^N \Bigr]
			=
			\\
			&\sum_{\substack{ w\in {\cW} \\ \ell(w)=N}}\,\, \int_0^T \int_0^{t_1} \dots \int_0^{t_{|w|-1}} \bE\Bigl[\bigl( \tilde \iota(w)1\bigr)\bigl(\whW_{t_1},\dots,\whW_{t_{|w|}},t_1,\dots,t_{|w|}\bigr) \Bigr] dt_{|w|} \dots dt_1.
		\end{align*}
	\end{theorem}
	\begin{proof}
		This theorem follows immediately from Proposition \ref{Prop_MomRep} and \ref{Thm_MomRep} by substituting $K,\cI,\iota$ with $\tilde K, \tilde{\cI},\tilde{\iota}$.
	\end{proof}
	
	To conclude Section \ref{Section:Representation}, we remark (without proof) that the operator $\tilde{\iota}$ has a similar representation to that of Equation \eqref{Eq_wRep}. 
	
	\begin{proposition}
		Fix $w\in \cW$ and let $m=|w|$. Then there exist a constant $C_w$ only depending on $w$ such that
		\begin{align*}
			(\tilde{\iota}(w)1)&(x_1,\dots,x_m,t_1,\dots,t_m)
			=C_w \sum_{\mathbf{l} \in \cL^w} \psi_{\mathbf{l}}(x_1,\dots,x_m) \tilde{K}(t_{\alpha_{\mathbf{l}}(2)},t_2) \dots \tilde{K}(t_{\alpha_{\mathbf{l}}(m)},t_m).
		\end{align*}
	\end{proposition}
	
	\section{Weak error estimation}
	\label{Section:Estimation}
	
	Let $X_T$ be the stochastic process defined in Equation \eqref{eq:RoughVolatilityModelNew} and let $X_T^n$ be the discrete time approximation defined in Equation \eqref{eq:RoughVolatilityModel-Discrete}. Since the class of polynomial test function is invariant under the shift $\Phi (.) \leftrightarrow \Phi( X_0 + . )$ we can and will assume without loss of generality $X_0 = 0$. Recalling Equation \eqref{eq:Def_DiscWords_K}, we note that
	\begin{equation*}
		\whW_{\eta(s)}= \int_0^s K(\eta(s),r) \, dW_r = \int_0^s \tilde{K}(s, r) dW_r.
	\end{equation*}
	For $\mathbf{t} \in \Delta_m^\circ$ we define $\eta(\mathbf{t})=(\eta(t_1),\dots,\eta(t_m))$. Theorem \ref{Thm_MomRep}, Theorem \ref{Thm_MomRepDisr} and Proposition \ref{Lem_wRep} imply that
	\begin{align}
		\nonumber
		\cE_{n, N} &\coloneqq \bE\Bigl[ (X_T)^N \Bigr] - \bE\Bigl[ (X_T^n)^N \Bigr]
		\\
		\nonumber
		&=\sum_{\substack{w\in \cW \\ \ell(w)=N}} \, \, \int\limits_{\mathbf{t}\in \Delta_{|w|}^{\circ}} \bE\Bigl[\bigl( \iota(w)1-\tilde \iota(w)1\bigr)\bigl(\whW_{\mathbf{t}}, \mathbf{t} \bigr) \Bigr] dt_{|w|} \dots dt_1
		\\
		\label{Eq_MomForm}
		&= \sum_{\substack{w\in \cW \\ \ell(w)=N}} 
		\sum_{\mathbf{l}\in \cL^w} \, 
		\int \limits_{\mathbf{t}\in \Delta_{|w|}^{\circ}}
		\bE\Bigl[
		\psi_{\mathbf{l}}\bigl( \whW_{\mathbf{t}} \bigr) 
		\Bigr] \cdot 
		\prod_{i=2}^{|w|} K(t_{\alpha_{\mathbf{l}}(i)},t_i) 
		dt_{|w|} \dots dt_1
		\\\nonumber
		&\quad-\sum_{\substack{w\in \cW \\ \ell(w)=N}}
		\sum_{\mathbf{l}\in \cL^w} \, 
		\int \limits_{\mathbf{t}\in \Delta_{|w|}^{\circ}}
		\bE\Bigl[
		\psi_{\mathbf{l}}\bigl( \whW_{\eta(\mathbf{t})} \bigr) 
		\Bigr] \cdot 
		\prod_{i=2}^{|w|} \tilde{K}(t_{\alpha_{\mathbf{l}}(i)},t_i) 
		dt_{|w|} \dots dt_1.
	\end{align}
	
	We recall Equation \eqref{eq:Simplex} that $\Delta_m^{\circ}$ denotes the simplex and $\mathbf{t}\in \Delta_m^\circ$ for  $\mathbf{t}=(t_1,\dots,t_m)$. Using the triangle inequality, for each $w\in \cW$ such that $|w| = m$ and $\mathbf{l} \in \cL^{w}$, we define $\alpha_{\mathbf{l}}$ according to Definition \ref{Def_Afcts} and consider
	\begin{equation}
		\label{eq:Estimation1}
		\int\limits_{\mathbf{t}\in \Delta_{|w|}^{\circ}} \Bigg( 
		\bE\Bigl[ \psi_{\mathbf{l}}\bigl( \whW_{\mathbf{t}} \bigr) \Bigr] 
		- \bE\Bigl[
		\psi_{\mathbf{l}}\bigl( \whW_{\eta(\mathbf{t})} \bigr) 
		\Bigr] \Bigg) \cdot 
		\prod_{i=2}^{|w|} \tilde{K}(t_{\alpha_{\mathbf{l}}(i)},t_i) 
		dt_{|w|} \dots dt_1
	\end{equation}
	and for any $j\in \{2, ..., |w|\}$
	\begin{align}
		\label{eq:Estimation2}
		\int \limits_{\mathbf{t}\in \Delta_{|w|}^{\circ}}
		\mkern-10mu
		\bE\Bigl[
		\psi_{\mathbf{l}}\bigl( \whW_{\eta(\mathbf{t})} \bigr) 
		\Bigr]&\cdot
		\prod_{i=2}^{j-1} K(t_{\alpha_{\mathbf{l}}(i)},t_i)
        \\
        \nonumber
		&\cdot\Bigl(\!K(t_{\alpha_{\mathbf{l}}(j)}, t_j) - \tilde{K}(t_{\alpha_{\mathbf{l}}(j)}, t_j) \Bigr) \cdot \mkern -10mu
		\prod_{i=j+1}^{|w|} \mkern -10mu \tilde{K}(t_{\alpha_{\mathbf{l}}(i)},t_i) 
		dt_{|w|} \dots dt_1.
	\end{align}
	
	\pagebreak[2]
	
	\begin{definition}
		Let $m\in \bN$. For a function $\varphi:\Delta_m^{\circ} \to \bR$, we define 
		\begin{equation}
			\label{eq:TildePhi}
			\tilde \varphi(t_1,\dots,t_m)=\varphi(\eta(t_1),\dots,\eta(t_m)). 
		\end{equation}
		For $1\le k \le m$ we define
		\begin{gather}
			\label{eq:Delta_k-varphi}
			\begin{aligned}
				\Delta_k \varphi(t_1, \dots,t_m)&=\varphi\bigl(\eta(t_1),\dots,\eta(t_{k-1}),t_k,t_{k+1} \dots,t_m\bigr)
				\\
				&\qquad-\varphi\bigl(\eta(t_1),\dots,\eta(t_{k-1}),\eta(t_k),t_{k+1}\dots,t_m\bigr). 
			\end{aligned}
		\end{gather}
		Finally, let $w\in \cW$ such that $|w| = m$ and fix $\mathbf{l} \in \cL^w$. Let $\alpha_{\textbf{l}}:\{1, ..., m\} \to \{0, 1, ..., m\}$ according to Definition \ref{Def_Afcts} and define
		\begin{equation}
			\label{eq:frakturI}
			\fI_{k}^{\alpha}=\int_0^T \int_0^{t_1} \dots \int_0^{t_{m-1}} \bigl|\Delta_k\varphi(t_1,\dots,t_m)\bigr| \cdot \prod_{j=2}^m K(t_{\alpha(j)},t_{j}) \, dt_m \, \dots \, dt_1. 
		\end{equation}
	\end{definition}
	
	To commence, we introduce a set of assumptions for our function $\varphi: \Delta_m^{\circ} \to \bR$ that are weaker than the results established in Lemma \ref{Lem_VerifAss}. This set of assumptions turns out to be all that is necessary to prove the desired upper bounds in Section \ref{subsec:4.2} and \ref{subsec:4.3} so we include it here in order to state our results in the most generality. 
	
	\begin{assumption}
		\label{Ass_G}
		Let $\varphi: \Delta_m^\circ \to \bR$  and extend it as $0$ outside of the simplex.
		Let $1<k<m$ and assume that there is a constant $C_m$ such that
		\begin{align}
			\bigl|\partial_k \varphi(\mathbf{t}) \bigr|&\le C_m \Bigl(t_{k-1}^{-H}t_{k}^{2H-1} +t_{k-1}^{-H}(t_{k-1}-t_k)^{2H-1} + t_{k}^{-H}(t_{k}-t_{k+1})^{2H-1}\Bigr)\label{Ass_G1}
			\\
			\bigl|\partial_1 \varphi(\mathbf{t}) \bigr| &\le C_m t_{1}^{-H}(t_{1}-t_2)^{2H-1} \label{Ass_G2}
			\\
			\bigl|\partial_m \varphi(\mathbf{t}) \bigr|&\le C_m \Bigl(t_{m-1}^{-H}t_{m}^{2H-1} +t_{m-1}^{-H}(t_{m-1}-t_m)^{2H-1} \Bigr)\label{Ass_G3}
			\\
			\bigl|\varphi(\mathbf{s})-\varphi(\mathbf{t} ) \bigr| &\le C_m  \|s-t\|^{2H}. \label{Ass_G4}
		\end{align}
	\end{assumption}
	
	\subsection{Estimates for Equation \eqref{eq:Estimation1}}
	\label{subsec:4.2}

	\begin{lemma}
		\label{Lem_VolFunc}
		Let $m \in \mathbb{N}$ and let $\varphi:\Delta_m^{\circ} \to \bR$. Let $w\in \cW$ such that $|w| = m$ and let $\textbf{l} \in \cL^{w}$. Let $1\le k \le m$ and let
		$\fI_{k}^{\alpha}$ according to Equation \eqref{eq:frakturI}. 
		
		Suppose that $\varphi$ satisfies Assumption \ref{Ass_G}. Then
		\begin{align*}
			\fI_k^{\alpha}\lesssim \begin{cases}
				C_m 2^m T^{m(H+1/2)}\Bigl( n^{-3H-1/2} \vee n^{-1}\Bigr) & \text{if } H \not = 1/6,
				\\
				C_m 2^m T^{m(H+1/2)}  n^{-1}\log(n) & \text{if } H  = 1/6.
			\end{cases}
		\end{align*}
		Here the hidden constant does not depend on $m,n,\alpha$ or $\varphi$, but may depend on $H$ and $T$.
	\end{lemma}
	
	\begin{proof}
		Note that for any choice of $\textbf{l} \in \cL^{w}$
		\begin{align*}
			\fI_{k}^{\alpha} \le  \fJ_{k}:=\int_{\mathbf{t}\in \Delta_m^\circ} \bigl| \Delta_k\varphi(\mathbf{t}) \bigr| \cdot \prod_{j=2}^{m}K(t_{j-1},t_{j}) \, dt_m \, \dots \, dt_1,
		\end{align*}
		We treat the cases $k=m$ and $k<m$ separately. For the form, we split the integral into three cases
		\begin{align*}
			\fJ_{m}^{(1)}&=\int \limits_{\substack{\mathbf{t}\in \Delta_m^\circ \\ t_m <2/n}} \bigl|\Delta_m\varphi(\mathbf{t})\bigr| \cdot \prod_{j=2}^{m}K(t_{j-1},t_{j}) \, dt_m \, \dots \, dt_1,
			\\
			\fJ_{m}^{(2)}&=\int \limits_{\substack{\mathbf{t}\in \Delta_m^\circ \\ t_m >t_{m-1}-2/n}} \bigl|\Delta_m\varphi(\mathbf{t})\bigr| \cdot \prod_{j=2}^{m}K(t_{j-1},t_{j}) \, dt_m \, \dots \, dt_1,
			\\
			\fJ_{m}^{(3)}&=\int \limits_{\substack{\mathbf{t}\in \Delta_m^\circ \\ 2/n<t_m <t_{m-1}-2/n}} \bigl|\Delta_m\varphi(\mathbf{t})\bigr| \cdot \prod_{j=2}^{m}K(t_{j-1},t_{j}) \, dt_m \, \dots \, dt_1.
		\end{align*}

		Recalling Equation \eqref{eq:Delta_k-varphi}, by \eqref{Ass_G4} in Assumption \ref{Ass_G} we have $\Delta_m\varphi \le C_m n^{-2H}$. In $\fJ_1^{(1)}$ we can estimate $K(t_{m-1},t_{m})\lesssim K(2/n,t_m)$ and then integrate $t_m$ from $0$ to $t_{m-1}$. By integrating the inner integrals we get
		\begin{gather}
			\begin{aligned}
				\label{Eq_FormRate1}
				\fJ_{m}^{(1)}&+\fJ_{m}^{(2)}
				\\
				& \lesssim C_m n^{-3H-1/2} \int_0^T \int_0^{t_1}   \dots \int_0^{t_{m-2}} K(t_{1},t_{2})\ldots K(t_{m-2},t_{m-1}) \, dt_{m-1} \, \dots \, dt_1 
				\\
				&\lesssim C_m n^{-3H-1/2} \int_0^T \int_0^{t_1}   \dots \int_0^{t_{m-3}} K(t_{1},t_{2})\ldots K(t_{m-3},t_{m-2})t_{m-2}^{H+1/2} \, dt_{m-2} \, \dots \, dt_1 .
				\\
				&
				\lesssim C_m n^{-3H-1/2} 2^mT^{m(H+1/2)}. 
			\end{aligned}
		\end{gather}
		Note that here we see that the boundary terms contribute significantly and therefore the analysis cannot be omitted!
		
		For the last integral we use \eqref{Ass_G3} in Assumption \ref{Ass_G} to conclude that 
		\begin{equation*}
			|\Delta_m \varphi|\le C_m n^{-1}\Bigl( t_{m-1}^{-H}t_{m}^{2H-1} +t_{m-1}^{-H}(t_{m-1}-t_m)^{2H-1}\Bigr).
		\end{equation*}
		Therefore, for $H\not=1/6$ 
		\begin{gather}
			\begin{aligned}
				\label{Eq_FormRate2}
				\fJ_{m}^{(3)} &\lesssim C_m
				n^{-1}
				\mkern -20mu
				\int \limits_{\substack{\mathbf{t}\in \Delta_m^\circ \\ 2/n<t_m <t_{m-1}-2/n}}
				t_{m-1}^{-H}t_{m}^{2H-1} K(t_{1},t_{2})\ldots K(t_{m-1},t_{m}) \, dt_m \, \dots \, dt_1
				\\
				&+C_m n^{-1}
				\mkern -20mu
				\int \limits_{\substack{\mathbf{t}\in \Delta_m^\circ \\ 2/n<t_m <t_{m-1}-2/n}} t_{m-1}^{-H}(t_{m-1}-t_{m})^{2H-1} K(t_{1},t_{2})\ldots K(t_{m-1},t_{m}) \, dt_m \, \dots \, dt_1
				\\
				&\lesssim C_m
				n^{-1}\int_0^T \int_0^{t_1} \dots \int_0^{t_{m-2}} t_{m-1}^{-H}t_{m-1}^{3H-1/2} K(t_{1},t_{2})\ldots K(t_{m-2},t_{m-1}) \, dt_{m-1} \, \dots \, dt_1
				\\
				&+C_m n^{-1}\int_0^T \int_0^{t_1} \dots \int_0^{t_{m-2}} t_{m-1}^{-H}n^{-3H+1/2} K(t_{1},t_{2})\ldots K(t_{m-2},t_{m-1}) \, dt_{m-1} \, \dots \, dt_1
				\\
				&\lesssim n^{-3H-1/2} C_m 2^mT^{m(H+1/2)}.
			\end{aligned}
		\end{gather}
		The same calculation for $H=1/6$ yields
		\begin{gather}
			\begin{aligned}\label{Eq_H16calc}
				\fJ_{m}^{(3)} &\lesssim 
				C_m
				n^{-1}\log(n)\int \limits_{\substack{\mathbf{t}\in \Delta_{m-1}^\circ} } t_{m-1}^{-H}t_{m-1}^{3H-1/2} K(t_{1},t_{2})\ldots K(t_{m-2},t_{m-1}) \, dt_{m-1} \, \dots \, dt_1
				\\
				&\quad+C_m n^{-1}\int \limits_{\substack{\mathbf{t}\in \Delta_{m-1}^\circ} }
				t_{m-1}^{-H}\log(n) K(t_{1},t_{2})\ldots K(t_{m-2},t_{m-1}) \, dt_{m-1} \, \dots \, dt_1
				\\
				&\lesssim n^{-1}\log(n) C_m 2^mT^{m(H+1/2)}.
			\end{aligned}
		\end{gather}
		
		Now let $k<m$ and for the following calculations set $t_0=T$. We split $\fJ_k$ up in 
		\begin{align*}
			\fJ_{k}^{(1)}&=\int \limits_{\substack{\mathbf{t}\in \Delta_m^\circ \\ t_k <2/n}}
			\bigl|\Delta_k\varphi(\mathbf{t})\bigr| \cdot \prod_{j=2}^m K(t_{j-1},t_{j}) \, dt_m \, \dots \, dt_1,
			\\
			\fJ_{k}^{(2)}&=\int \limits_{\substack{\mathbf{t}\in \Delta_m^\circ \\ t_k >t_{k-1}-2/n}}
			\bigl|\Delta_k\varphi(\mathbf{t})\bigr| \cdot \prod_{j=2}^m K(t_{j-1},t_{j}) \, dt_m \, \dots \, dt_1,
			\\
			\fJ_{k}^{(3)}&=\int \limits_{\substack{\mathbf{t}\in \Delta_m^\circ \\ 2/n<t_k <t_{k-1}-2/n \\ t_{k+1}< 2/n }}
			\bigl|\Delta_k\varphi(\mathbf{t})\bigr|  \cdot \prod_{j=2}^m K(t_{j-1},t_{j}) \, dt_m \, \dots \, dt_1,
			\\
			\fJ_{k}^{(4)}&=\int \limits_{\substack{\mathbf{t}\in \Delta_m^\circ \\ 2/n<t_k <t_{k-1}-2/n \\ t_{k+1}>t_k-2/n }}
			\bigl|\Delta_k\varphi(\mathbf{t})\bigr|  \cdot \prod_{j=2}^m K(t_{j-1},t_{j}) \, dt_m \, \dots \, dt_1,
			\\
			\fJ_{k}^{(5)}&=\int \limits_{\substack{\mathbf{t}\in \Delta_m^\circ \\ 2/n<t_k <t_{k-1}-2/n \\ 2/n<t_{k+1} <t_{k}-2/n}}
			\bigl|\Delta_k\varphi(\mathbf{t})\bigr|  \cdot \prod_{j=2}^m K(t_{j-1},t_{j}) \, dt_m \, \dots \, dt_1.
		\end{align*}
		
		By \eqref{Ass_G4} in Assumption \ref{Ass_G} we have $\Delta_k\varphi \lesssim n^{-2H}$. By integrating all the kernels as before we see that
		\begin{align*}
			\fJ_{k}^{(1)}+\fJ_{k}^{(2)}+\fJ_{k}^{(3)}+\fJ_{k}^{(4)}
			\lesssim n^{-3H-1/2}  C_m 2^mT^{m(H+1/2)}.
		\end{align*}
		
		For the last integral we use \eqref{Ass_G1}, or \eqref{Ass_G2}, in Assumption \ref{Ass_G} implying that 
		\begin{equation*}
			|\Delta_k \varphi|\le C_m n^{-1}\Bigl( t_{k-1}^{-H}t_{k}^{2H-1} +t_{k-1}^{-H}(t_{k-1}-t_k)^{2H-1} + t_{k}^{-H}(t_{k}-t_{k+1})^{2H-1}\Bigr).
		\end{equation*}
		
		We start by integrating up to $t_k$ or $t_{k+1}$ to get that
		
		\begin{align*}
			\fJ_{k}^{(5)}&\lesssim n^{-1}C_m2^{m-k}T^{(m-k)(H+1/2)}
			\mkern -40mu
			\int \limits_{\substack{\mathbf{t}\in \Delta_k^\circ\\ 2/n<t_k <t_{k-1}-2/n}}
			\mkern -40mu
			t_{k-1}^{-H}t_{k}^{2H-1} \prod_{j=2}^k K(t_{j-1},t_{j}) \, dt_{k} \, \dots \, dt_1
			\\
			&+n^{-1}C_m2^{m-k}T^{(m-k)(H+1/2)}
			\mkern -40mu
			\int \limits_{\substack{\mathbf{t}\in \Delta_k^\circ \\ 2/n<t_k <t_{k-1}-2/n}}
			\mkern -40mu
			t_{k-1}^{-H}(t_{k-1}-t_k)^{2H-1}
			\prod_{j=2}^k K(t_{j-1},t_{j}) \, dt_{k} \, \dots \, dt_1
			\\
			&+n^{-1}C_m2^{m-k}T^{(m-k)(H+1/2)}
			\mkern -40mu
			\int \limits_{\substack{ \mathbf{t}\in \Delta_{k+1}^\circ \\ 2/n<t_k <t_{k-1}-2/n \\ 2/n<t_{k+1} <t_{k}-2/n}}
			\mkern -40mu
			t_{k}^{-H}(t_{k}-t_{k+1})^{2H-1}
			\prod_{j=2}^{k+1} K(t_{j-1},t_{j}) \, dt_{k+1} \, \dots \, dt_1
			\\
			&=\fJ_{k}^{(5,1)}+\fJ_{k}^{(5,2)}+\fJ_{k}^{(5,3)}.
		\end{align*}
		
		By the properties of the Beta distribution
		\begin{align*}
			\fJ_{k}^{(5,1)}&\lesssim  n^{-1}
			C_m2^{m-k}T^{(m-k)(H+1/2)}
			\int \limits_{\mathbf{t}\in \Delta_{k-1}^\circ }
			t_{k-1}^{-H}t_{k-1}^{3H-1/2} \prod_{j=2}^{k-1} K(t_{j-1},t_{j}) \, dt_{k-1} \, \dots \, dt_1
			\\
			&\lesssim
			n^{-1}C_m2^mT^{m(H+1/2)}.
		\end{align*}
		
		Simple integration with respect to $t_k$ tells us in the case $H\not=1/6$ that
		\begin{align*}
			\fJ_{k}^{(5,2)}& \lesssim n^{-1} C_m2^{m-k}T^{(m-k)(H+1/2)} \int \limits_{\mathbf{t}\in \Delta_{k-1}^\circ }
			t_{k-1}^{-H}n^{-3H+1/2} \prod_{j=2}^{k-1} K(t_{j-1},t_{j}) \, dt_{k-1} \, \dots \, dt_1
			\\
			&\lesssim n^{-3H-1/2}C_m2^mT^{m(H+1/2)}.
		\end{align*}
		If $H=1/6$ then
		\begin{align*}
			\fJ_{k}^{(5,2)}& \lesssim n^{-1} C_m2^{m-k}T^{(m-k)(H+1/2)} \int \limits_{\mathbf{t}\in \Delta_{k-1}^\circ }
			t_{k-1}^{-H}\log(n) \prod_{j=2}^{k-1} K(t_{j-1},t_{j}) \, dt_{k-1} \, \dots \, dt_1
			\\
			&\lesssim n^{-1}\log(n)C_m2^mT^{m(H+1/2)}.
		\end{align*}
		
		We see that $\fJ_{k}^{(5,3)}$ is equal to $\fJ_{k}^{(5,2)}$ with $k+1$ instead of $k$, therefore
		\begin{align*}
			\mathfrak{J}_{k}^{(5,3)} \lesssim 
			\begin{cases}
				n^{-3H-1/2}C_m 2^mT^{m(H+1/2)}, & \text{if } H\not=1/6, \\
				n^{-1}\log(n)C_m 2^mT^{m(H+1/2)}& \text{if } H=1/6.\end{cases}
		\end{align*}
	\end{proof}
	
	\subsection{Estimates for Equation \eqref{eq:Estimation2}}
	\label{subsec:4.3}
	
	Let $m>k\ge 1$. For this section we fix some $w \in \cW$ with $|w|=m$ and let $\cL^w$ from Definition \ref{Def_Afcts}.
	Furthermore, fix some $\mathbf{l} \in \cL^w$ and set $\alpha = \alpha_{\mathbf{l}}$.  Define the function 
	\begin{equation}
		\label{eq:G-function}
		G_{\alpha}(t_1,\dots,t_k)=\int_0^{t_k} \dots \int_0^{t_{m-1}} \tilde \varphi(t_1,\dots,t_m) \cdot \prod_{l=k+1}^m  \tilde K(t_{\alpha(l)},t_l) \, dt_m \dots \, dt_{k+1},
	\end{equation}
	with $\tilde \varphi$ defined in Equation \eqref{eq:TildePhi} and $\tilde K$ defined in Equation \eqref{eq:Def_DiscWords_K}.
	
	\pagebreak[3
 ]
	
	\begin{lemma}
		\label{Lem_GFunc}
		Suppose that $\varphi$ satisfies Assumption \ref{Ass_G} and let $G_{\alpha}:\Delta_{m}^{\circ} \to \bR$ be defined as in Equation \eqref{eq:G-function}. For $t_1>\dots>t_{k-1}>t_k>s_k$ we have
		\begin{gather}
			\begin{aligned}
				\label{Eq_HolderG}
				|G(t_1,\dots,t_{k-1},t_{k})&-G(t_1,\dots,t_{k-1},s_k)|
				\\
				&\lesssim mC_m 2^{m-k}T^{(m-k)(H+1/2)} \Bigl(n^{-2H}\vee |t_k-s_k|^{2H}\Bigr).
			\end{aligned}
		\end{gather}
		as well as 
		\begin{align}\label{Eq_HolderG2}
			|G_{\alpha}(t_1,\dots,t_{k-1},t_k)|\lesssim C_m 2^{m-k}T^{(m-k)(H+1/2)}.
		\end{align}
		Here the hidden constants do not depend on $t_1, \dots, t_{k-1}$ $m,\alpha$ or $\varphi$, but may depend on $H$ and $T$. 
	\end{lemma}
	
	\begin{proof}
		Let 
		\begin{gather*}
			\Delta \tilde \varphi=\tilde \varphi(t_1,\dots,t_{k-1},t_k,t_{k+1},\dots,t_m)-\tilde\varphi(t_1,\dots,t_{k-1},s_k,t_{k+1},\dots,t_m).
		\end{gather*}  
		By Definition of $\tilde K$ and the triangle inequality it follows that
		\begin{align*}
			|G(t_1,&\dots,t_{k-1},t_{k})-G(t_1,\dots,t_{k-1},s_k)|
			\\
			\le&\int_0^{s_k} \dots \int_0^{t_{m-1}} |\Delta\tilde \varphi| \prod_{l=k+1}^m  \tilde K(t_{l-1},t_l) \, dt_m \dots \, dt_{k+1}
			\\
			&+\sum_{j=k+1}^m\int_0^{s_k} \dots\int_0^{t_{m-1}} \|\tilde \varphi\|_\infty |\tilde K(t_k,t_j)-\tilde K(s_k,t_j)|\prod_{\substack{l=k+1\\l \not=j}}^m  \tilde K(t_{l-1},t_l) \, dt_m \dots \, dt_{k+1}
			\\
			&+\int_{s_k}^{t_k}\int_0^{t_{k+1}} \dots \int_0^{t_{m-1}} \|\tilde \varphi\|_\infty \prod_{l=k+1}^m  \tilde K(t_{l-1},t_l) \, dt_m \dots \, dt_{k+1}
			\\
			\eqqcolon& \fJ^{1}+\fJ^2+\fJ^3.
		\end{align*}
		
		By \eqref{Ass_G4} in Assumption \ref{Ass_G} for $\varphi$ and the definition of $\tilde \varphi$ it follows that
		\begin{align*}
			\fJ^1 \lesssim C_m 2^{m-k}T^{(m-k)(H+1/2)}  \Bigl(n^{-2H}\vee |t_k-s_k|^{2H}\Bigr).
		\end{align*}
		
		We look at $\fJ^2$. By integrating from $t_m$ to $t_{j+1}$ the remaining integral is of the form 
		\begin{align*}
			\int_0^{s_k} \tilde K(s_k,r)-\tilde K(t_k,r) \, dr&=\frac 1 {H+1/2}\Bigl(\eta(s_k)^{H+1/2}+(\eta(t_k)-s_k)_+^{H+1/2} -\eta(t_k)^{H+1/2}\Bigr) 
			\\
			&\lesssim n^{-H-1/2}\vee |t_k-s_k|^{H+1/2}.
		\end{align*}
		Note that by \eqref{Ass_G4} in Assumption \ref{Ass_G} implies that $\|\tilde \varphi\|_\infty\lesssim C_m$.
		By integrating over the remaining variables and summing up we see that
		\begin{gather*}
			\fJ^2 \lesssim mC_m 2^{m-k}T^{(m-k)(H+1/2)} \Bigl(n^{-2H}\vee |t_k-s_k|^{2H}\Bigr).
		\end{gather*}
		A similar computation shows \ref{Eq_HolderG2}.
		Finally, it follows from integration that
		\begin{gather*}
			\fJ^3 \lesssim C_m 2^{m-k}T^{(m-k)(H+1/2)}  \Bigl(n^{-2H}\vee |t_k-s_k|^{2H}\Bigr).
		\end{gather*}
	\end{proof}
	
	\begin{lemma}
		\label{Lem_KernelEstim}
		We define
		\begin{equation}
			\label{eq:tildeK-K}
			\Delta K(t, s) = \tilde{K}\bigl( t, s\bigr) - K\bigl( t, s\bigr).
		\end{equation}
		
		It holds for any $0 \le \beta < 1/2 -H $ or $1/2-H<\beta$ that
		\begin{align}
			\sup_{u\in [t,T]} \int_0^t \bigl|\Delta K(u,s)\bigr| \cdot (t-s)^\beta \, ds \lesssim n^{-H-1/2-\beta}\vee n^{-1},\label{Eq_KernelEstim1}
			\\
			\sup_{t\in [0,T]} \biggl| \int_0^t \Delta K(t,s) \, ds \biggr| \lesssim n^{-H-1/2}\wedge t^{H-1/2}n^{-1}.\label{Eq_KernelEstim2}
		\end{align}
		Further, for $\beta = 1/2 -H $ it holds that
		\begin{align*}
			\sup_{u\in [t,T]} \int_0^t \bigl|\Delta K(u,s)\bigr| \cdot (t-s)^\beta \, ds \lesssim n^{-1}\log(n).
		\end{align*}
	\end{lemma}
	
	\begin{remark}
		We remind the reader that in the statement of Theorem \ref{Thm_Rate}, we have a separate rate of convergence when $H=1/6$. This is due to integrals of the form $\int_0^{t-2/n} (t-s)^{3H-3/2}\, ds$, because there a logarithmic term appears if $H=1/6$. These integrals appear in Lemma \ref{Lem_KernelEstim} which is used for Lemma \ref{Lem_DiffKernel} as well as in Equation  \eqref{Eq_H16calc}.
	\end{remark}
	
	\begin{proof}
		See \cite[Lemma 2.1]{Gassiat2022Weak}. We only prove \eqref{Eq_KernelEstim1}, which is a slight modification of \eqref{Eq_KernelEstim2} in \cite{Gassiat2022Weak}. We split the integral into two parts
		\begin{align*}
			\int_0^t \bigl|\Delta K(u,s)\bigr|(t-s)^\beta \, ds=\int_0^{t-2/n} \bigl|\Delta K(u,s)\bigr|(t-s)^\beta \, ds+\int_{t-2/n}^t \bigl|\Delta K(u,s)\bigr|(t-s)^\beta \, ds.
		\end{align*}
		For the second integral note that
		\begin{align*}
			\int_{t-2/n}^{t} \bigl|\Delta K(u,s)\bigr| \cdot (t-s)^\beta \, ds
			&\lesssim  n^{-\beta} 
			\int_{t-2/n}^{t} \bigl(K(u,s)+K(\eta(u),s)\bigr) \, ds
			\\
			&\lesssim n^{-\beta}  \int_0^{1/n}K(1/n,s) \, ds\lesssim n^{-H-1/2-\beta}.
		\end{align*}
		Note that for $s\le u\le v$ we have $|K(v,s)-K(u,s)|\lesssim (v-u) (u-s)^{H-3/2}$, therefore, for $\beta \not=1/2-H$
		\begin{align*}
			\int_{0}^{t-2/n} \bigl|\Delta K(u,s)\bigr| \cdot (t-s)^\beta \, ds &\lesssim n^{-1} \int_{0}^{t-2/n} \bigl(t-s)^{\beta} (u-s)^{H-3/2} \underbrace{\Bigl(\frac{\eta(u)-s}{u-s}\Bigr)^{H-3/2}}_{\lesssim 1}\, ds
			\\
			&\lesssim n^{-1} \int_{0}^{t-2/n} \bigl(t-s)^{\beta+H-3/2} \, ds \lesssim n^{-H-1/2-\beta}.
		\end{align*}
		If $\beta=1/2-H$, then the same calculation shows that
		\begin{align*}
			\int_{0}^{t-2/n} \bigl|\Delta K(u,s)\bigr| \cdot (t-s)^\beta \, ds &\lesssim n^{-1}  \log(n)
		\end{align*}
	\end{proof}
	
	\begin{lemma}
		\label{Lem_DiffKernel}
		Let $m \in \mathbb{N}$ and let $w\in \cW$ such that $|w| = m$. For the fixed $\textbf{l} \in \cL^w$, we define $\alpha_{\textbf{l}}$ according to Definition \ref{Def_Afcts}. For $2\le k \le m$, we define
		\begin{gather}\label{eq:Lem_DiffKernel-1}
			\begin{aligned}
				\Delta_k K_{\alpha}&(t_1 ,\dots,t_m)=
				\\
				&\Big( \prod_{l=2}^{k-1}  K(t_{\alpha(l)},t_{l}) \Big) \cdot  \Big( K(t_{\alpha(k)},t_{k})-\tilde K (t_{\alpha(k)},t_k) \Big) \cdot \Big( \prod_{l=k+1}^m  \tilde K(t_{\alpha(l)},t_l) \Big).
			\end{aligned}
		\end{gather}
		Let $\varphi:\Delta_m^{\circ} \to \bR$ and (recalling Equation \eqref{eq:TildePhi}) define
		\begin{equation}
			\label{eq:Lem_DiffKernel-2}
			\fL_{k,\alpha}=\int\limits_{\mathbf{t}\in \Delta_{m}^{\circ}} \tilde \varphi(\mathbf{t}) \cdot \Delta_k K_{\alpha}(\mathbf{t}) \, dt_m \, \dots \, dt_1.
		\end{equation}
		
		Assume that $\varphi$ satisfies Assumption \ref{Ass_G}. Then
		\begin{align*}
			\mathfrak{L}_{k, \alpha}\lesssim 
			\begin{cases}
				mC_m 2^m T^{m(H+1/2)}\Bigl( n^{-3H-1/2} \vee n^{-1}\Bigr) & \text{if } H \not = 1/6,
				\\
				mC_m 2^m T^{m(H+1/2)}  n^{-1}\log(n) & \text{if } H  = 1/6.
			\end{cases}
		\end{align*}
		Here the hidden constant does not depend on $m,n,\alpha$ or $\varphi$, but may depend on $H$ and $T$.
	\end{lemma}
	
	\begin{proof}
		Recalling the notation of Equation \eqref{eq:G-function}, we have
		\begin{gather*}
			\fL_{k, \alpha}=\int_{\mathbf{t}\in \Delta_k^\circ}  G_{\alpha}(t_1,\dots,t_k) \cdot \Delta_k K(t_1,\dots,t_k) \, dt_k \, \dots \, dt_1.
		\end{gather*}    
		We split this up to
		\begin{align*}
			\fL_{k, \alpha}^{(1)}&=\int \limits_{\substack{\mathbf{t}\in \Delta_k^\circ
					\\ t_k<2/n}}  G\big( t_{1},\dots,t_{k-1},0 \big) \cdot  \Delta_k K(t_1,\dots,t_k) \, dt_k \, \dots \, dt_1,
			\\
			\fL_{k, \alpha}^{(2)}&=
			\int \limits_{\substack{\mathbf{t}\in \Delta_k^\circ
					\\ t_k<2/n}} \Bigl( G(t_{1},\dots,t_{k-1},t_{k})-G(t_{1},\dots,t_{k-1},0) \Bigr) \cdot \Delta_k K(t_1,\dots,t_k) \, dt_k \, \dots \, dt_1,
			\\
			\fL_{k, \alpha}^{(3)}&=
			\int \limits_{\substack{\mathbf{t}\in \Delta_k^\circ
					\\ t_k>2/n}} G\big( t_{1},\dots,t_{k-1},t_{k-1} \big) \cdot \Delta_k K(t_1,\dots,t_k) \, dt_k \, \dots \, dt_1,
			\\
			\fL_{k, \alpha}^{(4)}&=
			\int \limits_{\substack{\mathbf{t}\in \Delta_k^\circ
					\\ t_k>2/n}} \Bigl( G(t_{1},\dots,t_{k-1},t_{k})- G(t_{1},\dots,t_{k-1},t_{k-1}) \Bigr) \cdot \Delta_k K(t_1,\dots,t_k) \, dt_k \, \dots \, dt_1,
		\end{align*}
		For notational convenience let $\mathfrak{C}_{k,m}=C_m 2^{m-k}T^{(m-k)(H+1/2)}$.  It follows from the definition of $G_{\alpha}$ that $\fL_{k, \alpha}^{(1)}=0$.
		
		For $\fL_{k, \alpha}^{(2)}$ we use \eqref{Eq_HolderG} to get
		\begin{align*}
			\bigl|\fL_{k,\alpha}^{(2)}\bigr| &\lesssim 
			m \fC_{k,m}
			\int \limits_{\substack{\mathbf{t}\in \Delta_k^\circ
					\\ t_k<2/n}} n^{-2H} 
			K(t_{\alpha(k-1)},t_{k-1})\dots K(t_{\alpha(2)},t_{2}) \, dt_k \, \dots \, dt_1
			\\
			&\lesssim n^{-3H-1/2} m C_m 2^{m}T^{m(H+1/2)} .
		\end{align*}
		Here we see again that the boundary terms contribute significantly and therefore the analysis of them should not be omitted!
		
		We now use $\mathbf{t}=(t_1,\dots,t_{k-1})$. For $\fL_k^{(3)}$ we notice that the inner two integrals are of the form
		\begin{align*}
			\fM^{(3)}:=\int_{2/n}^{t_{k-2}} \int_{2/n}^{t_{k-1}} \Delta K(t_{\alpha(k)},t_k) \, dt_k \, K(t_{\alpha(k-1)},t_{k-1}) G(\mathbf{t}) \,dt_{k-1}.
		\end{align*}
		By integrating out we see
		\begin{align*}
			&\fM^{(3)}=\int_{2/n}^{t_{k-2}} \Bigl((t_{\alpha(k)}-2/n)^{H+1/2}-\bigl( \eta(t_{\alpha(k)})-2/n\bigr)^{H+1/2} \Bigr)   K(t_{\alpha(k-1)},t_{k-1}) G(\mathbf{t}) \,dt_{k-1}
			\\
			&- \int_{2/n}^{t_{k-2}} \Bigl((t_{\alpha(k)}-t_{k-1})^{H+1/2}-\bigl( \eta(t_{\alpha(k)})-t_{k-1}\bigr)^{H+1/2} \Bigr)   K(t_{\alpha(k-1)},t_{k-1}) G(\mathbf{t}) \,dt_{k-1}.
		\end{align*}
		If $t_{k-1}<4/n$ we rearrange, use Holder continuity of the power function and see that
		\begin{align*}
			\bigl|&\fM^{(3)}\bigr|\le \int_{2/n}^{4/n\wedge t_{k-2}} \Bigl|\Bigl(t_{\alpha(k)}-\frac2n\Bigr)^{H+1/2}
			\mkern -10mu -(t_{\alpha(k)}-t_{k-1})^{H+1/2} \Bigr|   K(t_{\alpha(k-1)},t_{k-1}) G(\mathbf{t}) \,dt_{k-1}
			\\
			&+ \int_{2/n}^{4/n\wedge t_{k-2}} \Bigl|\Bigl( \eta(t_{\alpha(k)})-\frac2n\Bigr)^{H+1/2}
			\mkern -10mu -\bigl( \eta(t_{\alpha(k)})-t_{k-1}\bigr)^{H+1/2} \Bigr|   K(t_{\alpha(k-1)},t_{k-1}) G(\mathbf{t}) \,dt_{k-1}
			\\
			&\lesssim
			\int_{2/n}^{4/n\wedge t_{k-2}} (t_{k-1}-2/n)^{H+1/2} K(t_{\alpha(k-1)},t_{k-1}) G(\mathbf{t}) \,dt_{k-1}
			\\
			&\lesssim \int_{2/n}^{4/n\wedge t_{k-2}} n^{-H-1/2}  K(4/n,t_{k-1}) G(\mathbf{t}) \,dt_{k-1}
			\\
			&\lesssim \fC_{k,m}n^{-1}.
		\end{align*}
		In the case $t_{k-1}>4/n$ we use Taylor's formula and split up the integral once more to get
		\begin{align*}
			\bigl|\fM^{(3)}\bigr|&\lesssim
			\int_{4/n}^{t_{k-2}} n^{-1} (t_{\alpha(k)}-2/n)^{H-1/2}   K(t_{\alpha(k-1)},t_{k-1}) G(\mathbf{t}) \,dt_{k-1}
			\\
			&+\int_{4/n}^{t_{k-2}} n^{-H-1/2}\wedge \Bigl(n^{-1} (\eta(t_{\alpha(k)})-t_{k-1})^{H-1/2}\Bigr)   K(t_{\alpha(k-1)},t_{k-1}) G(\mathbf{t}) \,dt_{k-1}
			\\
			&\lesssim 
			\fC_{k,m}\int_{4/n}^{t_{k-2}} n^{-1} (t_{k-1}-2/n)^{H-1/2}   K(t_{k-2},t_{k-1}) \,dt_{k-1}
			\\
			&\quad+\fC_{k,m}\int_{4/n}^{t_{k-2}} n^{-H-1/2}\wedge \Bigl(n^{-1} (\eta(t_{k-2})-t_{k-1})^{H-1/2}\Bigr)   K(t_{k-2},t_{k-1})  \,dt_{k-1}
			\\
			&\lesssim n^{-1}\fC_{k,m}+\fC_{k,m}\int_{4/n}^{t_{k-2}-2/n}  n^{-1} (t_{k-2}-t_{k-1})^{H-1/2}  K(t_{k-2},t_{k-1})  \,dt_{k-1}
			\\
			&\quad+\fC_{k,m}\int_{t_{k-2}-2/n}^{t_{k-2}} n^{-H-1/2} K(t_{k-2},t_{k-1})  \,dt_{k-1}
			\\
			&\lesssim \fC_{k,m} n^{-1}.
		\end{align*}
		Hence it follows that
		\begin{align*}
			\bigl|\fL_{k, \alpha}^{(3)}\bigr|\lesssim C_m 2^{m}T^{m(H+1/2)} n^{-1}.
		\end{align*}
		
		For $\fL_{k, \alpha}^{(4)}$ we also use \eqref{Eq_HolderG} but this time in combination with Lemma \ref{Lem_KernelEstim} to see for $H \not=1/6$
		\begin{align*}
			\bigl|\fL_{k, \alpha}^{(4)}\bigr| &\lesssim m\fC_{k,m}\int_0^T \int \limits_{\substack{\mathbf{t}\in \Delta_k^\circ
					\\ t_k>2/n}} 
			(t_{k-1}-t_k)^{2H} \Bigl| \Delta K(t_{\alpha(k)},t_k) \Bigr| \prod_{j=2}^{k-1}K(t_{\alpha(j)},t_{j}) \, dt_k \, \dots \, dt_1
			\\
			&\lesssim n^{-3H-1/2} m\fC_{k,m} \int_0^T \int_0^{t_{1}} \dots \int_0^{t_{k-2}} 
			\prod_{j=2}^{k-1}K(t_{\alpha(j)},t_{j}) \, dt_{k-1} \, \dots \, dt_1
			\\
			&\lesssim n^{-3H-1/2} mC_m 2^{m}T^{m(H+1/2)} .
		\end{align*}
		If $H=1/6$ the same calculation yields
		\begin{align*}
			\bigl|\fL_{k, \alpha}^{(4)}\bigr| &\lesssim  n^{-1}\log(n) mC_m 2^{m}T^{m(H+1/2)}.
		\end{align*}
	\end{proof}
    
	\subsection{Examples}
	\label{subsec:4-Examples}

	Following on from Example \ref{example:Gassiat-X^4} and Example \ref{Ex_MomRep2}, we have that
	\begin{align*}
		\bE\Bigl[ (X_T)^4& \Bigr] -\bE\Bigl[ (X_T^n)^4 \Bigr] 
		\\
		=&
		6 \int_0^T \int_0^t \bE\Bigl[ f\bigl( \whW_s\bigr)^2 f\bigl( \whW_t\bigr)^2 \Bigr]-\bE\Bigl[ f\bigl( \whW_{\eta(s)}\bigr)^2 f\bigl( \whW_{\eta(t)}\bigr)^2 \Bigr] \, ds \, dt &\eqqcolon \fE^{(1)}
		\\
		&+24\rho^2 \int_0^T \int_0^t \int_0^s \bE\Bigl[ f\bigl( \whW_r \bigr) f \bigl( \whW_s \bigr) ff''\bigl( \whW_t \bigr) \Bigr] K(t,s) K(t,r)&
		\\
		&\qquad -\bE\Bigl[ f\bigl( \whW_{\eta(r)} \bigr) f \bigl( \whW_{\eta(s)} \bigr) ff''\bigl( \whW_{\eta(t)} \bigr) \Bigr] \tilde K(t,s) \tilde K(t,r) \, dr\,ds\,dt& \eqqcolon \fE^{(2)}
		\\
		&+24\rho^2 \int_0^T \int_0^t \int_0^s \bE\Bigl[ f\bigl( \whW_r \bigr) f \bigl( \whW_s \bigr) f'f'\bigl( \whW_t \bigr) \Bigr] K(t,s) K(t,r)&
		\\
		&\qquad -\bE\Bigl[ f\bigl( \whW_{\eta(r)} \bigr) f \bigl( \whW_{\eta(s)} \bigr) f'f'\bigl( \whW_{\eta(t)} \bigr) \Bigr]\tilde K(t,s) \tilde K(t,r) \, dr\,ds\,dt& \eqqcolon \fE^{(3)}
		\\
		&+
		24\rho^2 \int_0^T \int_0^t \int_0^s \bE\Bigl[ f\bigl( \whW_r \bigr) f'\bigl( \whW_s \bigr) ff'\bigl( \whW_t \bigr) \Bigr] K(t,s)K(s,r)&
		\\
		&\qquad -\bE\Bigl[ f\bigl( \whW_{\eta(r)} \bigr) f'\bigl( \whW_{\eta(s)} \bigr) ff'\bigl( \whW_{\eta(t)} \bigr) \Bigr] \tilde K(t,s)\tilde K(s,r) \, dr\,ds\,dt&\eqqcolon \fE^{(4)}
	\end{align*}

	We start with $\fE^{(1)}$ and see that
	\begin{align*}
		\fE^{(1)}&=6 \int_0^T \int_0^t \bE\Bigl[ f\bigl( \whW_s\bigr)^2 f\bigl( \whW_t\bigr)^2 \Bigr]-\bE\Bigl[ f\bigl( \whW_{\eta(s)}\bigr)^2 f\bigl( \whW_{t}\bigr)^2 \Bigr] \, ds \, dt
		\\
		&+6 \int_0^T \int_0^t \bE\Bigl[ f\bigl( \whW_{\eta(s)}\bigr)^2 f\bigl( \whW_t\bigr)^2 \Bigr]-\bE\Bigl[ f\bigl( \whW_{\eta(s)}\bigr)^2 f\bigl( \whW_{\eta(t)}\bigr)^2 \Bigr] \, ds \, dt
		\\
		&=\fI^{(1)}+\fI^{(2)}.
	\end{align*}
	Lemma \ref{Lem_VolFunc} implies that $\bigl|\fI^{(1)}\bigr|+\bigl|\fI^{(2)}\bigr|\lesssim n^{-3H-1/2}$.
	
	\pagebreak[3]
	
	We now consider $\fE^{(2)}$ and for simplicity let $H<1/6$. We write
	\begin{align*}
		\fE^{(2)}&=\int_0^T \int_0^t \int_0^s \bE\Bigl[ \Bigl(f\bigl( \whW_r \bigr)-f\bigl( \whW_{\eta(r)} \bigr)\Bigr) f \bigl( \whW_s \bigr) f'f'\bigl( \whW_t \bigr) \Bigr] K(t,s) K(t,r) \, dr\,ds\,dt
		\\
		&+
		\int_0^T \int_0^t \int_0^s \bE\Bigl[ f\bigl( \whW_{\eta(r)} \bigr)\Bigl( f \bigl( \whW_s \bigr)-f \bigl( \whW_{\eta(s)} \bigr) \Bigr) f'f'\bigl( \whW_t \bigr) \Bigr] K(t,s) K(t,r) \, dr\,ds\,dt
		\\
		&+
		\int_0^T \int_0^t \int_0^s \bE\Bigl[ f\bigl( \whW_{\eta(r)} \bigr)f \bigl( \whW_{\eta(s)} \bigr) \Bigl( f'f'\bigl( \whW_t \bigr)-f'f'\bigl( \whW_{\eta(t)} \bigr)\Bigr) \Bigr] K(t,s) K(t,r) drdsdt
		\\
		&+\int_0^T \int_0^t \int_0^s \bE\Bigl[ f\bigl( \whW_{\eta(r)} \bigr)f \bigl( \whW_{\eta(s)} \bigr)f'f'\bigl( \whW_{\eta(t)} \bigr)\Bigr] \Bigl(K(t,s)-\tilde K(t,s)\Bigr) K(t,r) \, dr\,ds\,dt
		\\
		&+\int_0^T \int_0^t \int_0^s \bE\Bigl[ f\bigl( \whW_{\eta(r)} \bigr)f \bigl( \whW_{\eta(s)} \bigr)f'f'\bigl( \whW_{\eta(t)} \bigr)\Bigr] \tilde K(t,s)\Bigl(K(t,r)-\tilde K(t,r)\Bigr) \, dr\,ds\,dt
		\\
		&=\fI^{(1)}+\fI^{(2)}+\fI^{(3)}+\fL^{(1)}+\fL^{(2)}
	\end{align*}
	Now Lemma \ref{Lem_VolFunc} implies that $\bigl|\fI^{(1)}\bigr|+\bigl|\fI^{(2)}\bigr|+\bigl|\fI^{(3)}\bigr|\lesssim n^{-3H-1/2}$. Furthermore Lemma \ref{Lem_DiffKernel} implies that $\bigl|\fL^{(1)}\bigr|+\bigl|\fL^{(2)}\bigr|\lesssim n^{-3H-1/2}$. Combining these estimates we showed that $\bigl|\fE^{(2)}\bigr|\lesssim n^{-3H-1/2}$. Repeating this procedure for $\fE^{(3)}$ and $\fE^{(4)}$ we see that
	\begin{equation*}
		\Bigl|\bE\Bigl[ (X_T)^4 \Bigr] -\bE\Bigl[ (X_T^n)^4 \Bigr] \Bigr|\lesssim n^{-3H-1/2}.
	\end{equation*}

    \begin{remark}\label{Corr_Gen_Rate}
        Let $f\in C^\infty$ be such that for all $N \in \mathbb{N}$ $f$ and its first $N$ derivatives have exponential growth with constants $C_{f,N}'$ and $C_{f,N}$. Let $\Phi$ be harmonic of the form $\Phi(x)=\sum_{N=0}^\infty a_N x^N$ and such that
        \begin{align*}
            \sum_{N=0}^\infty |a_N| C_N <\infty,
        \end{align*}
        for $C_N$ as in Remark \ref{Rem_Rate}. 
        Then
        \begin{align*}
            \Bigl|\bE\bigl[\Phi\bigl(X_T^n\bigr)-\Phi(\bigl(X_T\bigr)\bigr]\Bigr| \lesssim 
            \begin{cases}
    			n^{-3H-1/2} \vee n^{-1} & \text{for } H \not = 1/6,
    				\\
    			n^{-1}\log(n) & \text{for } H  = 1/6.
    			\end{cases}
        \end{align*}
        One useful example might be if $f=1+\operatorname{tanh}$, because in this case each derivative is bounded and therefore one could choose $C_{f,N}=0$, which significantly increases the choices of test functions $\Phi$. 
    \end{remark}

	\section{Proof of main Theorems}
	\label{Section:MainProof}

     In this Section, we consolidate the previous results to prove Theorem \ref{Thm_Rate}:
     \newline
    In Section \ref{Section:Representation}, we established an explicit formula for the moments of the stochastic process Equation \eqref{eq:RoughVolatilityModelNew} and the discrete time approximation \eqref{eq:RoughVolatilityModel-Discrete}. We use this formula to expand the error as in Equation \eqref{Eq_MomForm} and proceed by looking at each term individually.
    For these differences we use the upper bounds established in Section \ref{Section:Estimation}, most importantly Lemma Lemma \ref{Lem_VolFunc} and Lemma \ref{Lem_DiffKernel}.
	
	\subsection[Proof of Theorem 1.1, part (i)]{Proof of Theorem \ref{Thm_Rate}, part \ref{Thm1i}}
	
	\begin{proof}
		In this proof we only treat the case $H\not=1/6$. The case $H=1/6$ works exactly the same way, except there are logarithmic correction terms when applying Lemma \ref{Lem_VolFunc} and Lemma \ref{Lem_DiffKernel}.
		
		We start by considering Equation \eqref{Eq_MomForm}: let $w \in \cW$ such that $\ell(w)=N$ and $|w|=m$. Next, fix $\textbf{l} \in \cL^w$ and let $\alpha = \alpha_{\mathbf{l}}$ as in Definition \ref{Def_Afcts}. By Proposition \ref{Lem_wRep}, it suffices to consider the difference of integrals of the form
		\begin{equation*}
			C_w\int\limits_{\mathbf{t}\in \Delta_{m}^{\circ}}
			\Bigg( \bE\Bigl[
			\varphi(\mathbf{t})
			\Bigr] \cdot 
			\prod_{i=2}^{m} K(t_{\alpha(i)},t_i) 
			- \bE\Bigl[
			\tilde \varphi(\mathbf{t})
			\Bigr] \cdot 
			\prod_{i=2}^{m} \tilde{K}(t_{\alpha(i)},t_i) 
			\Bigg) dt_{m} \dots dt_1
		\end{equation*}
		where
		\begin{equation*}
			\varphi(t_1,\dots,t_m)=\bE\Bigl[\psi_{\textbf{l}}(\whW_{t_1},\dots,\whW_{t_m})\Bigr] \quad \mbox{and}\quad \tilde\varphi(t_1,\dots,t_m)=\bE\Bigl[\psi_{\textbf{l}}(\whW_{\eta(t_1)},\dots,\whW_{\eta(t_m}))\Bigr].
		\end{equation*}
		Proposition \ref{Lem_wRep} tells us that the function $\psi_{\textbf{l}}$ is of the form \eqref{eq:Lem_wRep} so we only need to take derivatives up to $N-2$ times because words $w$ with $\ell(w)=N$ can contain the letter $I$ only up to $N-1$ times. Hence $\psi_{\textbf{l}} \in C^2$. Equation \eqref{eq:Lem_wRep} also already implies that $\psi_{\textbf{l}}$ is of at most exponential growth with constants $C_{\psi_{\textbf{l}}}'=2^m (C_f')^{2m}$ and $C_{\psi_{\textbf{l}}}=2C_f$: There are $m$ derivatives giving rise to $2^m$ terms and each individual term is a function of product structure $g(x_1,\dots,x_m)=\prod_{k=1}^m g_k(x_k)$ , where each $g_k$ is again of exponential growth.  
		
		Furthermore, there are at most $\#\{w \in \cW:\ell(w)=N\}<2^N$ words and for each $w \in \cW$ with $\ell(w)=N$. Note that if $|w|=m$ it follows that $\bigl|\cL^w\bigr|\le m!$. So in total there are at most $2^N N!$ of these terms time $N$ terms coming from triangle inequalities.
		
		Recall the definition for $\fI_{k}$ from Equation \eqref{eq:frakturI}. Courtesy of Lemma \ref{Lem_VerifAss} we can apply Lemma \ref{Lem_VolFunc} and get that 
		$$
		\fI_k \lesssim C_m 2^m T^{m(H+1/2)} \Big( n^{-1/2-3H} \vee n^{-1} \Big).
		$$
		From Lemma \ref{Lem_VerifAss} it follows that
        \begin{equation*}
             C_m \lesssim C'_{\psi_{\textbf{l}}} m \exp\Big( \frac{(C_{\psi_{\textbf{l}}} \cdot m)^2}{2} \cdot \frac{T^{2H+1}}{(2H+1/2)^2} \Big),
        \end{equation*}
		where the multiplicative constant is independent of $k,m,N$.
		
		Next, recalling $\Delta_k K_{\alpha}(t_1 ,\dots,t_m)$ defined in Equation \eqref{eq:Lem_DiffKernel-1}, we note that
		$$
		\Delta_k K(t_1,\dots,t_m) \le \Big( \prod_{l=2}^{k}  K(t_{l-1},t_{l}) \prod_{l=k+1}^m  \tilde K(t_{l-1},t_l)  \Big) + \Big( \prod_{l=2}^{k-1}  K(t_{l-1},t_{l}) \prod_{l=k}^m \tilde K(t_{l-1},t_l) \Big).
		$$
		It follows that
		\begin{align*}
			\int_0^T&\int_0^{t_1} \dots \int_0^{t_{m-1}}  \prod_{l=2}^{k}  K(t_{l-1},t_{l})\prod_{l=k+1}^m  \tilde K(t_{l-1},t_l) \, dt_m \, \dots \, dt_1
			\\
			&=
			\int_0^T\int_0^{t_1} \dots \int_0^{t_{m-2}} (t_{m-2}-t_{m-1})^{H-1/2} \int_0^{t_{m-1}}  (t_{m-1}-t_m)^{H-1/2} \, dt_m \, dt_{m-1}\dots \, dt_1
			\\
			&\lesssim T^m,
		\end{align*}
		which implies that $\Delta_k K$ is integrable.
		We look at terms of the form \eqref{eq:Lem_DiffKernel-2}: By Lemma \ref{Lem_VerifAss}  we can use Lemma \ref{Lem_DiffKernel} to see that 
		
		$$\fL_{k, \alpha} \lesssim 
		mC_m2^mT^{m(H+1/2)} \Bigl(n^{-1/2-3H} \vee n^{-1}\Bigr).$$
	
		Recalling the bound for $C_w$ in \ref{Rem_Cw} given by
		$$
		|C_w|=|\rho|^{2|w|-\ell(w)}2^{\ell(w)-|w|}\ell(w)! \le \ell(w)! \le (N!). 
		$$

	Putting everything together and applying triangle inequality we see that
	\begin{align*}
    	\Bigl|\bE\Bigl[ \Phi& \bigl( X_T \bigl) \Bigr] - \bE\Bigl[ \Phi  \bigl( X_T^{(n)} \bigl) \Bigr] \Bigr|
        \\
    	\lesssim &  
    	2^N N! 
    	\cdot N 
    	\cdot N! 
    	\cdot N2^NT^{N(H+1/2)}
    	\\
    	&\cdot
    	2^N\bigl(C'_f\bigr)^{2N} N \exp\Big( \frac{(2C_f \cdot N)^2}{2} \cdot \frac{T^{2H+1}}{(2H+1/2)^2} \Big) 
    	\Bigl(n^{-3H-1/2} \vee n^{-1}\Bigr) 
	\end{align*}
	\end{proof}

	\subsection[Proof of Theorem 1.1, part (ii)]{Proof of Theorem \ref{Thm_Rate}, part \ref{Thm1ii}}
	
	\begin{proof}
		We start by considering Equation \eqref{Eq_MomForm}. By Remark \ref{Rem_Cw} we see that $C_w\not =0$ if and only if $N$ is even and $w=\wJ^{m}$ with $m=N/2$. Using this fact, Definition \ref{definition:I-J_operations} as well as Proposition  \ref{Lem_wRep} we see that
		\begin{align*}
			\bE\Bigl[X_T^N\Bigr]-\bE\Bigl[\bigl(X_T^n\bigr)^N\Bigr]=\frac{N!}{2^m}\int \limits_{\mathbf{t}\in \Delta_{m}^{\circ}}
			\varphi(\mathbf{t})-\varphi\bigl(\eta(\mathbf{t})\bigr)\,dt_{m} \dots dt_1,
		\end{align*}
		where
		\begin{equation*}
			\varphi(\mathbf{t})=\bE\Bigl[f^2\bigl(\whW_{t_1}\bigr)\cdots f^2\bigl(\whW_{t_m}\bigr) \Bigr].
		\end{equation*}
		For $k=1,\dots,m$ we define $\Delta_k \varphi$ as in Equation \ref{eq:Delta_k-varphi}. Then
		\begin{align*}
			\bE\Bigl[X_T^N\Bigr]-\bE\Bigl[\bigl(X_T^n\bigr)^N\Bigr]=\sum_{k=1}^m \frac{N!}{2^m}\int \limits_{\mathbf{t}\in \Delta_{m}^{\circ}}
			\Delta_k\varphi(\mathbf{t})\,dt_{m} \dots dt_1.
		\end{align*}
		By Lemma \ref{Lem_VerifAss} and Taylor's formula we see that 
		\begin{align}\label{Eq_Bd0Corr}
			\Bigl| \Delta_k\varphi \Bigr| \lesssim \begin{cases}
				\frac1n(t_1-t_2)^{2H-1} & \text{if } k=1,
				\\
				\frac1n(t_k-t_{k+1})^{2H-1}+\frac1n(t_{k-1}-t_k)^{2H-1} & \text{if } 1<k<m,
				\\
				\frac1nt_m^{2H-1}+\frac1n(t_{m-1}-t_m)^{2H-1} & \text{if } k=m.
			\end{cases}
		\end{align}
		In any case the right hand side of \eqref{Eq_Bd0Corr} is an integrable function on $\Delta_m^\circ$ times the factor $1/n$. Using triangle inequality it therefore follows that
		\begin{align*}
			\Bigl|\bE\Bigl[X_T^N\Bigr]-\bE\Bigl[\bigl(X_T^n\bigr)^N\Bigr] \Bigr|\lesssim n^{-1}.
		\end{align*}
	\end{proof}
	
	\section{Lower Bound}
	\label{sec:LB}
	For this final section, we want to demonstrate a setting under which we have a lower bound for the weak error that is asymptotically the same as the upper bound established in Section \ref{Section:MainProof}. In order to simplify some of the challenging calculations, we restrict ourselves to the setting \cite{Gassiat2022Weak} and note that this work does not consider any lower bound type results. 
	
	\begin{proposition}
		\label{proposition:lowerbound}
		Let $H<1/6$ and consider the case when $f(x)=x$, $\rho=1$ and $\Phi(x)=x^3$, meaning we consider a weak error of the form $\cE_{n,3}=\bE\bigl[ (X_T)^3 \bigr] - \bE\bigl[ (X_T^n)^3 \bigr]$. Then it follows from numeric computations that
		\begin{equation*}
			\liminf_{n \rightarrow \infty} n^{3H+1/2} \cE_{n,3}>0.
		\end{equation*}
		In particular, this means that
		\begin{equation*}
			\cE_{n,3}\gtrsim n^{-3H-1/2}.
		\end{equation*}
	\end{proposition}
	
	In order to prove Proposition \ref{proposition:lowerbound}, we establish a truncation of the weak error into three deterministic integrals. Recall the covariance kernel of Liouville fractional Brownian motion $C(t,s)=\bE\bigl[\whW_{t}\whW_{s}\bigr]$ defined in Equation \eqref{eq:Covariance-Liouville}. Then the weak error is of the form
	\begin{equation*}
		C\cE_{n,3}=\int_0^1 \int_0^t K(t,s) \cdot C(t,s) \, ds\,dt-\int_0^T \int_0^t K(\eta(t),s) \cdot C(\eta(t),\eta(s)) \, ds\,dt
	\end{equation*}
	where $C$ is a renormalising constant dependent on $H$. We split this up into
	\begin{align*}
		C\cE_{n,3}=&\int_0^1 \int_0^t \bigl( K(t,s)-K(\eta(t),s) \bigr) \cdot C(t,t) \, ds\,dt
		&\quad =: \fA^{(1)}
		\\
		&+\int_0^1 \int_0^t \bigl(K(t,s)-K(\eta(t),s)\bigr) \cdot \bigl( C(t,s)-C(t,t) \bigr) \, ds\,dt
		&\quad =: \fA^{(2)}
		\\
		&+\int_0^1 \int_0^t K(\eta(t),s) \cdot \bigl( C(t,s) - C(\eta(t),\eta(s)) \bigr) \, ds\,dt
		&\quad =: \fA^{(3)}
	\end{align*}
	We will also regularly refer to the deterministic integrals
	\begin{align}
		\label{eq:C-integrals1}
		\fB_1=&-\int_0^{\infty} v^{H-1/2} \Bigl( v^{H-1/2}- (v-1)_+^{H-1/2} \Bigr) \,dv<0
		\\
		\label{eq:C-integrals2}
		\fB_2=&\int_0^\infty  v^{H-1/2}\Bigl(v^{2H}-(v+1)^{2H}\Bigr) \, dv<0
		\\
		\label{eq:C-integrals3}
		\fB_3=&(H-1/2)\int_0^{\infty} v^{H-1/2} (1+v)^{H-3/2} \,dv<0.
	\end{align}
	
	\subsection{Technical Lemmas}
	
	Firstly, we prove some upper bounds necessary for the proof of Proposition \ref{proposition:lowerbound}. 
	\begin{lemma}
		\label{Lem_Remainder}
		Let
		\begin{equation*}
			F(s,t)=C(s,t)-C(t,t)-(t-s)^{2H}\fB_1
		\end{equation*}
		Then
		\begin{equation*}   
			\int_0^1 \int_0^t \bigl((t-s)^{H-1/2}-(\eta(t)-s)_+^{H-1/2}\bigr) \cdot F(s,t) \, ds\,dt\lesssim n^{-1}.
		\end{equation*}
	\end{lemma}
	
	\begin{proof}
		First note that using the substitution $r=t-(t-s)v$ we see that
		\begin{align*}
			C(s,t)-C(t,t)&=\int_0^t (t-r)^{H-1/2} \cdot \Bigl((t-r)^{H-1/2}-(s-r)_+^{H-1/2} \Bigr) \, dr
			\\
			&=-(t-s)^{2H}\int_0^{\frac{t}{t-s}} v^{H-1/2} \cdot \Bigl(v^{H-1/2}-(v-1)_+^{H-1/2} \Bigr) \, dv.
		\end{align*}
		Thus it follows that
		\begin{gather*}
			F(s,t)=(t-s)^{2H}\int_{\frac{t}{t-s}}^\infty v^{H-1/2} \cdot \Bigl((v-1)_+^{H-1/2}-v^{H-1/2} \Bigr) \, dv.
		\end{gather*}
		Note that
		\begin{gather*}
			\int_{\frac{t}{t-s}}^\infty v^{H-1/2} \cdot \Bigl((v-1)_+^{H-1/2}-v^{H-1/2} \Bigr) \, dr\lesssim  \int_{\frac{t}{t-s}}^\infty v^{2H-2}\, dv \lesssim (t-s)^{1-2H}t^{2H-1}.
		\end{gather*}
		Now we split up
		\begin{align*}
			\int_0^1 \int_0^t& \bigl((t-s)^{H-1/2}-(\eta(t)-s)_+^{H-1/2}\bigr) \cdot F(s,t) \, ds\,dt 
			\\
			=& \int_0^{2/n}\int_0^t\bigl((t-s)^{H-1/2}-(\eta(t)-s)_+^{H-1/2}\bigr) \cdot F(s,t) \, ds\,dt & =: \fD_1
			\\
			&+\int_{2/n}^1 \int_{t-2/n}^t \bigl((t-s)^{H-1/2}-(\eta(t)-s)_+^{H-1/2}\bigr) \cdot F(s,t) \, ds\,dt
			& =: \fD_2
			\\
			&+\int_{2/n}^1 \int_0^{t-2/n}\bigl((t-s)^{H-1/2}-(\eta(t)-s)_+^{H-1/2}\bigr) \cdot F(s,t) \, ds\,dt
			& =: \fD_3
		\end{align*}
		For $\fD_1$ it follows from the boundedness of $F$ that
		\begin{equation*}
			\fD_1\lesssim \int_0^{2/n} \int_0^t (t-s)^{H-1/2}\,ds\,dt\lesssim n^{-1}.
		\end{equation*}
		For $\fD_2$ we see that
		\begin{equation*}
			\fD_2 \lesssim \int_{2/n}^1 \int_{t-2/n}^t (t-s)^{1-2H}t^{2H-1}(t-s)^{H-1/2}\,ds\,dt\lesssim n^{-3/2+H}.
		\end{equation*}
		For $\fD_3$ note that for $t-s>2/n$ we have that $\bigl|(t-s)^{H-1/2}-(\eta(t)-s)_+^{H-1/2}\bigr|\lesssim n^{-1}(t-s)^{H-3/2}$. Thus we see that
		\begin{equation*}
			\fD_3\lesssim n^{-1}\int_{2/n}^1 \int_0^{t-2/n} (t-s)^{1-2H}t^{2H-1}(t-s)^{H-3/2}\,ds\,dt\lesssim n^{3/2-H}\lesssim n^{-1}.
		\end{equation*}
	\end{proof}
	
	\subsection{Proof of Proposition \ref{proposition:lowerbound}}
	
	The first step is to show that the following Lemma holds:
	\begin{lemma}
		\label{Lem_I123}
		Let $0<H<1/6$. There are constants $C_2,C_3$ such that $0<C_3<C_2$ and such that
		\begin{align*}
			\fA^{(1)}&\lesssim n^{-1},
			\\
			\fA^{(2)}&\geq C_2 n^{-3H-1/2}+o\bigl(n^{-3H-1/2}\bigr),
			\\
			0 \geq\fA^{(3)}&\geq -C_3 n^{-3H-1/2}+o\bigl(n^{-3H-1/2}\bigr).
		\end{align*}
		The constants are explicitely given by
		\begin{equation}
			\label{eq:Lem_I123}
			C_2=\fB_1\fB_2 \frac{1}{3H+3/2}
			\quad \mbox{and} \quad
			C_3=-\fB_3 \frac12\Bigl(\frac{2^{3H-1/2}}{1/2-3H} +\frac{2}{1-2H} \Bigr).
		\end{equation}
	\end{lemma}
	
	\begin{proof}
		Note that \cite[Proof of Theorem 2.1, case (2)]{Gassiat2022Weak} shows that $|\fA^{(1)}|\lesssim n^{-1}$ . 
		
		\emph{Step 1. $\fA^{(2)}$}: Thanks to Lemma \ref{Lem_Remainder}, we have that
		\begin{equation*}
			\int_0^1 \int_0^t \bigl((t-s)^{H-1/2}-(\eta(t)-s)_+^{H-1/2}\bigr) \cdot F(s,t) \, ds\,dt\lesssim n^{-1}. 
		\end{equation*}
		Thus, 
		\begin{equation}
			\label{eq:Lem_I123-A2-1}
			\fA^{(2)}=\fB_1\int_0^1 \int_0^t \bigl((t-s)^{H-1/2}-(\eta(t)-s)^{H-1/2}\bigr) \cdot (t-s)^{2H} \, ds\,dt + o\bigl(n^{-3H-1/2}\bigr),
		\end{equation}
		where $\fB_1$ is defined by \eqref{eq:C-integrals1}. Using the substitution $s=t-(t-\eta(t))v$ it follows for the inner integral that
		\begin{align}
			\nonumber
			\int_0^t \bigl((t-s)^{H-1/2}-&(\eta(t)-s)^{H-1/2}\bigr)(t-s)^{2H} \, ds
			\\
			\nonumber
			&= (t-\eta(t))^{3H+1/2} \int_0^{\frac t {t-\eta(t)}} \bigl(v^{H-1/2}-(v-1)_+^{H-1/2}\bigr)v^{2H} \, dv
			\\
			\label{eq:Lem_I123-A2-2}
			&= (t-\eta(t))^{3H+1/2} \Bigl( \fB_2 - \int_{\frac t {t-\eta(t)}}^\infty  \bigl(v^{H-1/2}-(v-1)_+^{H-1/2}\bigr)v^{2H} \, dv \Bigr),
		\end{align}
		where $\fB_2$ is defined by \eqref{eq:C-integrals2}. 
		
		Next, we consider the function
		\begin{equation*}
			G(t)=(t-\eta(t))^{3H+1/2} \int_{\frac t {t-\eta(t)}}^\infty \bigl(v^{H-1/2}-(v-1)_+^{H-1/2}\bigr)v^{2H} \, dv.
		\end{equation*}
		For $t>2/n$, we have that $\frac{t}{t-\eta(t)}>2$ so that
		\begin{align}
			\nonumber
			\int_{\frac t {t-\eta(t)}}^\infty \bigl(v^{H-1/2}-(v-1)_+^{H-1/2}\bigr)v^{2H} \, dv &\le (1/2-H)\int_{\frac t {t-\eta(t)}}^\infty (v-1)^{H-3/2}v^{2H} \, dv
			\\
			\nonumber
			&\le (1/2-H)\int_{\frac t {t-\eta(t)}}^\infty (v-1)^{3H-3/2}\, dv
			\\
			\label{eq:Lem_I123-A2-4}
			&\lesssim t^{3H-1/2}\bigl(t-\eta(t)\bigr)^{-3H+1/2}.
		\end{align}
		Finally, we have that
		\begin{equation*}
			\sup_{t\in[0,2/n]} |G(t)| < n^{-3H-1/2} \int_{0}^\infty \bigl|v^{H-1/2}-(v-1)_+^{H-1/2}\bigr| \cdot v^{2H} \, dv<\infty, 
		\end{equation*}
		and thanks to Equation \eqref{eq:Lem_I123-A2-4}
		\begin{align*}
			\int_{2/n}^1 G(t) dt \lesssim \int_{2/n}^1 t^{3H-1/2} \bigl(t - \eta(t) \bigr) dt \lesssim n^{-1}. 
		\end{align*}
		Then we see that
		\begin{equation}
			\label{eq:Lem_I123-A2-3}
			\int_0^1 G(t) \,dt = \int_0^{2/n} G(t) \, dt +\int_{2/n}^1 G(t) \, dt \lesssim n^{-1}.
		\end{equation}
		
		Combining Equation \eqref{eq:Lem_I123-A2-1}, Equation \eqref{eq:Lem_I123-A2-2} and Equation \eqref{eq:Lem_I123-A2-3}
		\begin{align*}
			\fA^{(2)}&= \fB_1 \Bigl( \fB_2 \int_0^1 (t-\eta(t))^{3H+1/2} \,dt - \int_0^1 G(t) dt \Bigr)
			\\
			&\geq n^{-3H-1/2} \frac{1}{3H+3/2} \fB_1\fB_2+o\bigl(n^{-3H-1/2}\bigr).
		\end{align*}
		
		\emph{Step 2. $\fA^{(3)}$}: Recall that
		\begin{equation*}
			\fA^{(3)}=\int_0^1 \int_0^t K(\eta(t),s)\bigl(C(t,s)-C(\eta(t),\eta(s))\bigr) \, ds\,dt.
		\end{equation*}
		We write 
		\begin{equation*}
			C(t,s)-C(\eta(t),\eta(s))
			=
			\Bigl( C(t,s)-C(t,\eta(s)) \Bigr)
			+
			\Bigl( C(t,\eta(s))-C(\eta(t),\eta(s)) \Bigr).
		\end{equation*}
		Note that $C(t,s)-C(t,\eta(s))>0$ for all $T>t>s$. As we want a lower bound for the (negative) term $\fA^{(3)}$ we can disregard this term. 
		
		We now look at the function $t\mapsto C(t,s)$ for $t>s$ and see that
		\begin{align*}
			\partial_t C(t,s)&=(H-1/2)\int_0^s (s-r)^{H-1/2} (t-r)^{H-3/2}\, dr
			\\
			&=(H-1/2)(t-s)^{2H-1} \int_0^{\frac{s}{t-s}} v^{H-1/2} (1+v)^{H-3/2} \,dv
			\ge \fB_3 (t-s)^{2H-1},
		\end{align*}
		where we recall $\fB_3$ was defined in Equation \eqref{eq:C-integrals3}. Therefore it follows that
		\begin{align*}
			H(t)&\coloneqq \int_0^{\eta(t)} \Bigl( C(t,\eta(s))-C(\eta(t),\eta(s)) \Bigr) (\eta(t)-s)^{H-1/2} \, ds
			\\
			&\geq \fB_3 \int_0^{\eta(t)} (t-\eta(t)) (\eta(t)-\eta(s))^{2H-1} (\eta(t)-s)^{H-1/2}\, ds.
		\end{align*}
		Direct calculation yields that
		\begin{align*}
			\int_{\eta(t)-1/n}^{\eta(t)} (\eta(t)-\eta(s))^{2H-1} (\eta(t)-s)^{H-1/2}\, ds\, = n^{1-2H} \frac{1}{1-2H} n^{-H-1/2}.
		\end{align*}
		Going one step further we have
		\begin{align*}
			\int_{\eta(t)-2/n}^{\eta(t)-1/n} (\eta(t)-\eta(s))^{2H-1} (\eta(t)-s)^{H-1/2}\, ds\, \le n^{1-2H} \frac{1}{1-2H} n^{-H-1/2}.
		\end{align*}
		Finally we see that
		\begin{align*}
			&\int_{0}^{\eta(t)-2/n} (\eta(t)-\eta(s))^{2H-1} (\eta(t)-s)^{H-1/2}\, ds 
			\\
			&\le \int_{0}^{\eta(t)-2/n} (\eta(t)-s)^{3H-3/2} \, ds
			=\frac{1}{1/2-3H}\Bigl( \bigl(2/n\bigr)^{3H-1/2}-\bigl(\eta(t)\bigr)^{3H-1/2} \Bigr)
			\\
			&\le \frac{2^{3H-1/2}}{1/2-3H} n^{1/2-3H}.
		\end{align*}
		Putting these together we see that
		\begin{equation*}
			H(t)\geq \fB_3 \Bigl(\frac{2^{3H-1/2}}{1/2-3H} +\frac{2}{1-2H} \Bigr)n^{1/2-3H}\bigl(t-\eta(t)\bigr),
		\end{equation*}
		and therefore
		\begin{align*}
			\fA^{(3)}&\geq  \int_0^1 H(t) \, dt
			\geq \fB_3 \int_0^1  \Bigl(\frac{2^{3H-1/2}}{1/2-3H} +\frac{2}{1-2H} \Bigr)n^{1/2-3H}(t-\eta(t)) \, dt
			\\
			&\ge  \fB_3 \frac12\Bigl(\frac{2^{3H-1/2}}{1/2-3H} +\frac{2}{1-2H} \Bigr) n^{-3H-1/2}. 
		\end{align*}
		The fact that $C_2>C_3$ then follows from Lemma \ref{Lem_C2C3}.
	\end{proof}
	
	\begin{lemma}\label{Lem_C2C3}
		Let $C_2$ and $C_3$ be defined as in Equation \eqref{eq:Lem_I123} where $\fB_1$, $\fB_2$ and $\fB_3$ are as defined in Equations \eqref{eq:C-integrals1}, \eqref{eq:C-integrals2} and \eqref{eq:C-integrals3}. Then $C_2>C_3$. 
	\end{lemma}
	
	\begin{proof}
		Recall the Euler-Beta function
		\begin{equation*}
			B(z_1, z_2) = \int_0^1 v^{z_1 - 1} \cdot (1-v)^{z_2-1} dv
		\end{equation*}
		where $z_1, z_2 \in \bC$ such that $Re(z_1), Re(z_2)>0$. Further, this function can be extended holomorphically to the complex plane. Then the Euler-Beta function satisfies the well known identity
		\begin{equation*}
			B(z_1, z_2) = \frac{ \Gamma(z_1) \cdot \Gamma(z_2)}{\Gamma(z_1 + z_2) }
		\end{equation*}
		Using techniques from complex analysis one can show that    
		\begin{align*}
			\eqref{eq:C-integrals1} =& \fB_1 = 2^{-1-2H}B\Bigl(H+1/2,-H\Bigr)
			\\
			\eqref{eq:C-integrals2} =& \fB_2 = -B\Bigl(-3H- \tfrac{1}{2}, \tfrac{1}{2}+H\Bigr)
			\\
			\eqref{eq:C-integrals3} =& \fB_3 = \bigl( H - \tfrac{1}{2} \bigr) \cdot B\Bigl( 1-2H, H + \tfrac{1}{2} \Bigr)
		\end{align*}
		
		Hence
		\begin{align*}
			C_2 - C_3&= \fB_1\fB_2 \frac{1}{3H+3/2} + \fB_3 \Bigl(\frac{2^{3H-3/2}}{1/2-3H} +\frac{1}{1-2H} \Bigr)
			\\
			&=-2^{-1-2H}B\Bigl(H+1/2,-H\Bigr)B\Bigl(-3H-\tfrac{1}{2}, \tfrac{1}{2}+H\Bigr) \frac{1}{3H+3/2} 
			\\
			&\quad+ \Bigl(H-\frac12\Bigr)B\Bigl( 1-2H, H + \tfrac{1}{2} \Bigr) \Bigl(\frac{2^{3H-3/2}}{1/2-3H} +\frac{1}{1-2H} \Bigr). 
		\end{align*}
		
		\begin{figure}[!ht]
			\begin{centering}
				\includegraphics[width=0.9\columnwidth]{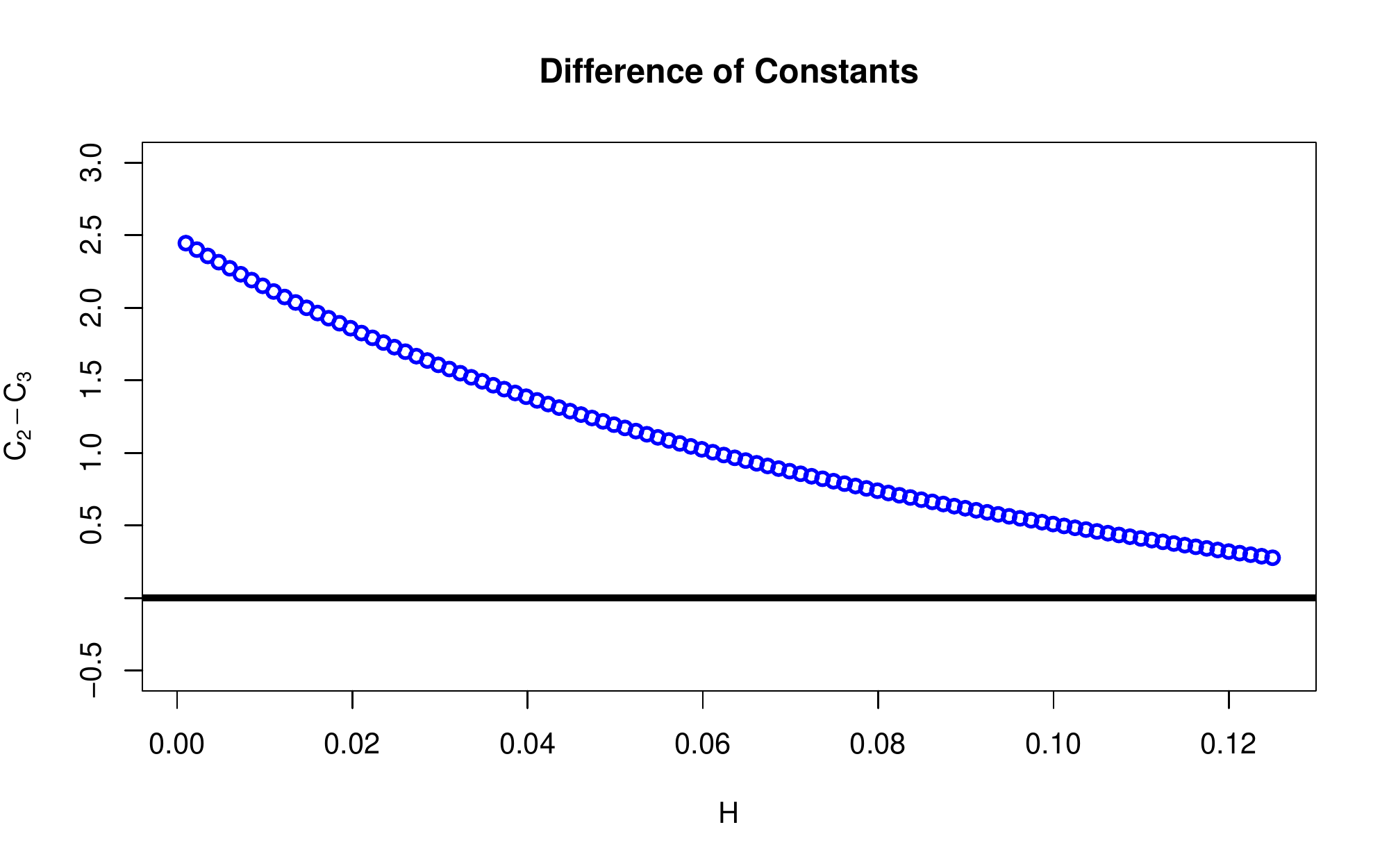}
				\caption{Numerical Computation of the difference $C_2-C_3$}
				\label{Fig_1}
			\end{centering}
		\end{figure}
		In particular, $H\mapsto C_2-C_3$ is a (locally) holomorphic map. A numerical simulation, see Figure \ref{Fig_1}, now shows that at least for $1/1000<H<1/8$ it holds that $C_2-C_3>0$.
	\end{proof}

	\begin{proof}[Proof of Proposition \ref{proposition:lowerbound}]
		We use the decomposition
		\begin{equation*}
			C\mathcal{E}_{n,3}=\fA^{(1)}+\fA^{(2)}+\fA^{(3)}.
		\end{equation*}
		By Lemma \ref{Lem_I123} it follows that
		\begin{align*}
			C\mathcal{E}_{n,3}&\ge C_2n^{-3H-1/2}-C_3n^{-3H-1/2}+o\bigl(n^{-3H-1/2}\bigr)
			\\
			&=(C_2-C_3)n^{-3H-1/2}+o\bigl(n^{-3H-1/2}\bigr).
		\end{align*}
		Thus it follows that
		\begin{gather*}
			\liminf_{n \rightarrow \infty}n^{3H+1/2}\mathcal{E}_{n,3}\ge \frac{C_2-C_3}{C}>0.
		\end{gather*}
	\end{proof}

\bibliographystyle{abbrvurl}
\bibliography{Bibliography.bib}

\end{document}